\documentclass[12pt, onecolumn]{IEEEtran}
\IEEEoverridecommandlockouts
\usepackage[utf8]{inputenc} 
\usepackage[T1]{fontenc}
\usepackage{url}
\usepackage{ifthen}
\usepackage{cite}
\usepackage[cmex10]{amsmath} 
\usepackage{amssymb}
\usepackage{graphicx}
\usepackage{epstopdf}
\usepackage{epsfig}
\usepackage{multirow}
\usepackage{amsthm}
\usepackage{comment}
\usepackage{caption}
\usepackage{enumitem}
\usepackage{amsmath}
\usepackage{color}
\usepackage{bbm}
\usepackage{algorithm}
\usepackage{algpseudocode}
\usepackage{csquotes}
\usepackage{subfiles}
\usepackage{caption, multirow, makecell}
\usepackage{array}
\usepackage{caption}
\usepackage{afterpage}
\usepackage{dblfloatfix}
\usepackage{hyperref}
\usepackage{booktabs}
\usepackage{tikz}
\usepackage{longtable}
\usepackage{blkarray}
\usepackage{adjustbox}
\newtheorem{thm}{Theorem}
\newtheorem{lem}{Lemma}

\newtheorem{cor}{Corollary}
\newtheorem{example}{Example}

\newtheorem{defn}{Definition}

\def\BibTeX{{\rm B\kern-.05em{\sc i\kern-.025em b}\kern-.08em
		T\kern-.1667em\lower.7ex\hbox{E}\kern-.125emX}}

\begin{document}
	
	\title{A General Plotkin-type Bound on Function-Correcting Codes with Wyner-Graham Distance}
\author{\IEEEauthorblockN{Kanchana Lokshmii Jagatti, K. Hareesh, N.T. Rashid Ummer and B. Sundar Rajan} \\
	\IEEEauthorblockA{\textit{Department of Electrical Communication Engineering} \\
		\textit{Indian Institute of Science,}
		Bangalore, India \\
		\{kanchanal, hareeshk, rashidummer, bsrajan\}@iisc.ac.in}
}

	\maketitle
	
	\begin{abstract}
		Function-correcting codes (FCCs) are designed to protect a specified function evaluation of messages at a higher level than the level of protection for messages, against errors while reducing the redundancy required for reliable communication. FCCs have thus far been studied for channels matched to various distances, including the Hamming and Lee distances. Every function partitions the message space into preimage sets corresponding to its distinct function values. Existing Plotkin-type bounds on the optimal redundancy of FCCs under the studied distances, applicable to arbitrary functions on the message space, depend on the pairwise distances among all the message vectors. This makes these bounds difficult to compute. We derive a general Plotkin-type bound on the optimal redundancy of FCCs under Wyner–Graham distances, which include the Hamming and Lee distances as special cases. Our bound depends only on the cardinalities of the preimage sets and the sum of pairwise distances only among vectors within each preimage set. This approach significantly reduces the computations required to evaluate the existing Plotkin-type bounds and yields simplified bounds that are easier to compute for specific functions. We obtain simplified Plotkin-type bound for linear functions under the Wyner-Graham distance. Furthermore, the existing simplified bounds for linear functions under the Hamming and Lee distances are recovered as special cases of the proposed bound. We also obtain simplified bounds for several important classes of functions, including the Hamming weight function, the Hamming weight distribution function, the monomial functions under the Hamming distance, and the modular sum function, the Lee weight function, and the Lee weight distribution function under the Lee distance.

	\end{abstract}

	
	\section{Introduction}
	
Error-correcting codes (ECCs) are traditionally designed to enable reliable recovery of an entire message transmitted over a noisy channel. In many communication scenarios the receiver is interested in recovering the complete message and also interestied  in evaluating a specific function of the message. Motivated by this observation, Lenz \textit{et al.} in \cite{LBWY} introduced the notion of \textit{function-correcting codes} (FCCs) in the Hamming distance setting. FCCs are designed to protect function evaluations of messages at a higher level than the level of protection for messages, against errors while reducing the redundancy required for reliable communication with only conventional ECCs. The authors in \cite{LBWY} showed that FCCs can achieve substantially lower redundancy than conventional ECCs when the function value is also required at the receiver. Furthermore, the authors established a connection between FCCs and \textit{irregular-distance codes}. Using this connection, general upper and lower bounds on the optimal redundancy of FCCs were derived. These bounds depend heavily on the function under consideration. Simplified bounds were also obtained for specific functions, including the Hamming weight function and the Hamming weight distribution function. 

Following the introduction of FCCs, several extensions and generalizations have been investigated. In \cite{PR}, the authors obtained a lower bound on the redundancy of FCCs using the concept of function-dependent graph and further obtained a simplified lower bound for linear functions. An upper bound for the binary field, staying within a logarithmic factor of the known lower bound in \cite{LBWY}, was presented in \cite{LS}. Improved lower bound and upper bound on optimal redundancy for the Hamming weight function were derived in \cite{ZXZG}. FCCs for a class of functions that assume only a limited number of values within a given Hamming ball, referred to as locally bounded functions, were studied in \cite{CSFH}. In \cite{CSFH2}, FCCs were generalized to \textit{function-correcting partition codes}, which are defined directly on partitions of the message space rather than on specific target functions. Another generalization of FCCs, termed \textit{FCCs with data protection}, was proposed in \cite{CSFH3}, which considers the simultaneous protection of both the message and its function value. FCCs with data protection for a specific class of function, called the Hamming code membership function, was investigated in \cite{SAARB,AAS}.

In \cite{LBWY, PR, LS, ZXZG, CSFH, CSFH2, CSFH3, SAARB, AAS}, FCCs were studied under the Hamming distance. The FCC framework has also been extended to several other channel models, including symbol-pair read channels \cite{XLC}, $b$-symbol read channels \cite{AAY, SR}, channels matched to homogeneous distance \cite{LL}, Lee distance channels \cite{GA, HRS}, and insertion-deletion channels \cite{AA}. A central problem in the study of FCCs under any distance framework is the characterization of the optimal redundancy required to protect a given function. Plotkin-type lower bounds on the optimal redundancy, applicable to arbitrary functions on the message space, were obtained in \cite{LBWY, XLC, AAY, LL, HRS, AA} for the corresponding distance frameworks. These bounds are expressed in terms of a matrix named the distance requirement matrix (DRM) in \cite{LBWY} associated with the function. Since the DRM depends on the distance requirements between all the pairs of message vectors corresponding to distinct function values, its computation generally requires computing pairwise distances of a large number of vector pairs. Consequently, evaluating the resulting bounds becomes computationally intensive as the size of the message space grows. A simplified lower bound on the redundancy of FCCs for linear functions under the Hamming distance was derived in \cite{PR}. Subsequently, similar bounds for linear functions under the $b$-symbol read channel and the Lee distance channel were obtained in \cite{SR} and \cite{GA}, respectively. Simplified lower bounds for specific functions have been derived in some of the existing works; however, only for certain parameter regimes. For instance, lower bounds on the optimal redundancy of FCCs for the Hamming weight function were derived in \cite{LBWY, ZXZG}, but only for $k > t$, where $k$ denotes the message length and $t$ denotes the number of channel errors to be corrected. For the Lee weight function, a lower bound was derived in \cite{HRS} for $q \ge 5$ and $t=\lfloor \frac{q-3}{2} \rfloor$, while another lower bound was obtained in \cite{GA} for $k > \lceil \frac{t+1}{\lfloor \frac{q}{2} \rfloor} \rceil$, where $q$ denotes the alphabet size. Therefore, deriving simplified and computationally tractable bounds for arbitrary functions under various distance frameworks is of significant interest.

Every function naturally induces a partition of the message space into preimage sets corresponding to its distinct function values. This partition provides a useful representation of the function for studying function correction. The main idea of this paper is to exploit this partition to derive Plotkin-type bounds that depend only on the cardinalities of the preimage sets and the sums of pairwise distances among vectors within each preimage set, rather than on pairwise distances between vectors across different preimage sets. Two important distances that have been studied in the context of FCCs are the Hamming distance and the Lee distance. Both belong to a broader class of distances satisfying specific conditions introduced by Wyner and Graham in \cite{WG}. Motivated by this observation, we consider FCCs under this general distance framework, referred to in this paper as the \textit{Wyner-Graham distance} framework. The formal definition is provided in Section \ref{sec_prelim}. This framework includes the Hamming and Lee distances as special cases and enables a unified treatment of FCCs under these distances. Within this framework, we derive a general Plotkin-type bound on the optimal redundancy of FCCs by utilizing the partition of the message space induced by an arbitrary function.

The main contributions of this paper are summarized as follows:
\begin{itemize}
	\item We introduce a general Plotkin-type bound on the optimal redundancy of FCCs under the Wyner-Graham distance framework,  applicable to arbitrary functions. This bound depends on the cardinalities of the preimage sets and the corresponding sums of pairwise distances among vectors belonging to the same preimage set, thereby reducing the computational complexity of the existing Plotkin-type bounds. Since the Wyner-Graham distance framework includes the Hamming and Lee distances as special cases, specialized Plotkin-type bounds on the optimal redundancy of FCCs for these distances are obtained.
	\item We obtain simplified Plotkin-type bound for linear functions under the Wyner-Graham distance. We show that existing simplified bounds for linear functions under the Hamming and Lee distances are recovered as special cases of the proposed Plotkin-type bound.
	\item Using the proposed Plotkin-type bound, we derive simplified bounds for several important classes of functions, including the Hamming weight function, the Hamming weight distribution function, the monomial functions under the Hamming distance, the modular sum function, the Lee weight function, and the Lee weight distribution function  under the Lee distance. Unlike existing bounds specific for these functions, the proposed bounds are applicable in all parameter regimes. 
\end{itemize}

\subsection{Organization}
The rest of the paper is organized as follows. In Section \ref{sec_prelim}, we first define the Wyner-Graham distance and then generalize the concepts and definitions of FCCs introduced in \cite{LBWY} under the Hamming distance to the Wyner-Graham distance. In Section \ref{sec_Plotkin_WG}, we derive our general Plotkin-type bound for FCCs with Wyner-Graham distance. Also, the obtained Plotkin-type bound is specialized to the Hamming distance and the Lee distance.  The simplified Plotkin-type bound for linear functions under the Wyner-Graham distance is discussed in Section \ref{sec_Plotkin_linear}. In Section \ref{sec_Plotkin_Ham}, we obtain simplified bounds for specific functions under the Hamming distance including the Hamming weight function, the Hamming weight distribution function and the monomial functions. In Section \ref{sec_Plotkin_Lee}, the Plotkin-type bound is simplified to functions defined under the Lee distance such as the Lee weight function, the Lee weight distribution function and the modular sum function. Section \ref{sec_concl} concludes the paper.

\noindent \textit{Notations:}
Let $\mathbb{N}_0$ denote the set of non-negative integers and $\mathbb{N}_{0}^{M\times M}$ represents the set of all  $M\times M$ matrices with non-negative integer entries. For a matrix $\mathbf{D}\in\mathbb{N}_{0}^{M\times M}$, we denote its $(i,j)-th$  entry by $[\mathbf{D}]_{ij}$. For an integer $M$, we write $ [M]^+ \triangleq max\{M,0\}$ and for any positive integer $N$, $[N]$ denotes the set $\{1,2,...,N\}$. For any set $\mathcal{A}$, $|\mathcal{A}|$ denotes the cardinality of $\mathcal{A}$.  For sets $\mathcal{A}$ and $\mathcal{B}$, $\mathcal{A} \backslash \mathcal{B}$ denotes the elements in $\mathcal{A}$ but not in $\mathcal{B}$. For integers $a$ and $b$, $a \mid b$ reads as $a$ divides $b$.  For a prime power $q$,  $\mathbb{F}_q$ denotes the finite field with $q$ elements, $\mathbb{F}_q^{\ast}=\mathbb{F}_q \backslash \{0\}$ and $\mathbb{F}_q[x]$ denotes the ring of polynomials over $\mathbb{F}_q$.

\section{Preliminaries}\label{sec_prelim}
In this section, we first define the Wyner-Graham distance and state some of its properties. Then in the following subsection, we generalize the concepts and definitions of FCCs introduced in \cite{LBWY} under the Hamming distance to the Wyner-Graham distance.

\subsection{Wyner-Graham distance}
In \cite{WG}, Wyner and Graham obtained an upper bound on the maximun attainable minimum distance for a $q$-ary code for a certain class of distances. This class of distances includes the Hamming distance and the Lee distance as specific examples. In this paper, we refer to this class of distances as \textit{Wyner-Graham distance}. The formal definition is given below.
\begin{defn}[Wyner-Graham distance]\label{WG_dis_defn}
Let $\mathcal{X}$ be a finite alphabet set.	A distance $d_X: \mathcal{X} \times \mathcal{X} \to \mathbb{N}_0$ between two elements $x,y \in \mathcal{X}$, denoted by $d_X(x,y)$,  is called a Wyner-Graham distance if it satisfies the following conditions. 
	\begin{align*}
		&i. \text{ } d_X(x,x)=0,\\
		&ii. \text{ } d_X(x,y)=d_X(y,x),\\
		&iii. \text{ } \sum_{y \in \mathcal{X}} d_X(x,y)= S_X \text{ for all } x \in \mathcal{X}, \text{ where $S_X$ is some fixed number}, \\
		&iv.  \text{ For a fixed element } a\in \mathcal{X}, \text{ the array } \mathbf{Q}=(q_{x,y})_{x,y \in \mathcal{X} \backslash \{a\}}, \text{  is non-negative definite, } \\& \text{ where } q_{x,y}=d_X(a,x)+d_X(a,y)-d_X(x,y). 
	\end{align*}
	 Given two vectors $ \mathbf{x}=(x_1,x_2,...,x_n), \mathbf{y}=(y_1,y_2,...,y_n) \in \mathcal{X}^n$, the Wyner-Graham distance $d_X: \mathcal{X}^n \times \mathcal{X}^n \to \mathbb{N}_0$ is defined as $d_X(\mathbf{x},\mathbf{y})\triangleq \displaystyle \sum_{i=1}^{n} d_X(x_i,y_i)$. 	
\end{defn}
The Hamming distance and the Lee distance satisfy the conditions of the Wyner-Graham distance \cite{WG}; hence, they are special cases of Wyner-Graham distance. To establish the notations, we review the definitions of the Hamming distance and the Lee distance below.
\begin{defn}[Hamming distance]
	Let $\mathcal{X}$ be a finite alphabet set.  The Hamming distance $d_H$ is defined as
 $$
d_H(x,y)=
\begin{cases}
	0, & x=y,\\
	1, & x\neq y,
\end{cases}
\qquad \text{ for all } x,y \in \mathcal{X}.
$$
\end{defn}
Let $\mathbb{F}_q$ denote the finite field of order $q$. From the definition of Hamming distance, it can be seen that  $\displaystyle \sum_{y\in\mathbb{F}_q} d_H(x,y)=q-1,
 \forall x\in\mathbb{F}_q. $ That is, the parameter $S_X$ appearing in Definition \ref{WG_dis_defn} equals $S_H=q-1$. 
\begin{defn}[Lee Weight and Lee Distance \cite{CW}]
For $x\in\mathbb{Z}_q$, the Lee weight is defined as $w_{L}(x) = min(x, q - x)$. The Lee distance between two symbols $x,y \in \mathbb{Z}_{q}$ is defined as  $$d_{L}(x,y) = min(|x - y|, q - |x - y|).$$ 
\end{defn}
For every $x\in\mathbb{Z}_q$, $\displaystyle \sum_{y\in\mathbb{Z}_q} d_L(x,y)=S_L,$ is given in \cite{WG} as
\begin{equation}\label{eq:lee}
	S_L=
	\begin{cases}
		\frac{q^2}{4}, & q \text{ is even},\\
		\frac{q^2-1}{4}, & q \text{ is odd}.
	\end{cases}
\end{equation}
The following Lemma will be used in the proof of the general Plotkin-type bound for optimal redundancy of FCCs with Wyner-Graham distance in Section \ref{sec_Plotkin_WG}.
\begin{lem}\label{lem_WG}
	For $ \mathbf{x}=(x_1,x_2,...,x_n), \mathbf{y}=(y_1,y_2,...,y_n) \in \mathcal{X}^n$,  we have 
	\begin{equation}
		 \sum_{\mathbf{x}, \mathbf{y} \in \mathcal{X}^n} d_{X}(\mathbf{x}, \mathbf{y})=n|\mathcal{X}|^{2n-1}S_X.
	\end{equation}
\end{lem}
\begin{IEEEproof}
\begin{align}\label{eq:xy}
	\sum_{\mathbf{x}, \mathbf{y} \in \mathcal{X}^n} d_{X}(\mathbf{x}, \mathbf{y}) & = \sum_{\mathbf{x}, \mathbf{y} \in \mathcal{X}^n} \sum_{i=1}^{n} d_{X}(x_i, y_i) \nonumber \\ & =  \sum_{i=1}^{n} \sum_{\mathbf{x}, \mathbf{y} \in \mathcal{X}^k} d_{X}(x_i, y_i) \nonumber \\ & = \sum_{i=1}^{n} |\mathcal{X}|^{2n-2} \sum_{x,y \in \mathcal{X}} d_X(x,y)  \nonumber \\ & =\sum_{i=1}^{n} |\mathcal{X}|^{2n-2} |\mathcal{X}| \sum_{y \in \mathcal{X}} d_X(x,y)  \nonumber \\ & = \sum_{i=1}^{n} |\mathcal{X}|^{2n-2} |\mathcal{X}| S_X \nonumber \\ & = n|\mathcal{X}|^{2n-1}S_X \nonumber.
\end{align}
\end{IEEEproof}
 
Next, we generalize the concepts and definitions of FCCs introduced in \cite{LBWY} under the Hamming distance to the Wyner-Graham distance $d_X$. 
 \subsection{Function-correcting codes with Wyner-Graham distance $d_X$} 
Let $\mathbf{u} \in \mathcal{X}^k$ be a message and let $ f : \mathcal{X}^k \to Im(f) \triangleq \{ f(\mathbf{u}) : \mathbf{u} \in \mathcal{X}^k \} $ be a function computed on $ \mathbf{u} $ with expressiveness $ E = |Im(f)| \leq |\mathcal{X}|^k $. The message is encoded via an encoding function
$$
\mathrm{Enc} : \mathcal{X}^k \to \mathcal{X}^{k+r}, \quad \mathrm{Enc}(\mathbf{u}) = (\mathbf{u}, p(\mathbf{u})),
$$
where $p(\mathbf{u}) \in \mathcal{X}^r$ is the redundancy vector and $r$ is the redundancy. The resulting codeword $\mathrm{Enc}(\mathbf{u})$ is transmitted over a channel, resulting in a received vector $\mathbf{y} \in \mathcal{X}^{k+r}$ satisfying $ d_X(\mathrm{Enc}(\mathbf{u}), \mathbf{y}) \le t $.
The formal definition of function-correcting codes with the Wyner-Graham distance $d_X$ is given below.

\begin{defn}[Function-Correcting Codes with Wyner-Graham distance $d_X$]\label{FCC_defn}
	A systematic encoding function $\mathrm{Enc} : \mathcal{X}^k \to \mathcal{X}^{k+r}, \mathrm{Enc}(\mathbf{u})=(\mathbf{u},p(\mathbf{u}))$ defines a function-correcting code (FCC) for a function	$f:\mathcal{X}^k \to \operatorname{Im}(f)$ under the distance $d_X$ if for all $\mathbf{u}_1,\mathbf{u}_2 \in \mathcal{X}^k $ with $f(\mathbf{u}_1)\neq f(\mathbf{u}_2),$ we have $$d_X(\mathrm{Enc}(\mathbf{u}_1),\mathrm{Enc}(\mathbf{u}_2)) \ge 2t+1.$$
	
\end{defn}

By Definition \ref{FCC_defn}, for any received vector $\mathbf{y}$, which is obtained by at most $t$ errors from $\mathrm{Enc}(\mathbf{u})$, the receiver can uniquely recover $f(\mathbf{u})$, provided it has knowledge of the function $f$, the encoding function $\mathrm{Enc}$, and the distance $d_X$. Next, we define the term \textit{optimal redundancy} of an FCC under the Wyner-Graham distance $d_X$.

\begin{defn}[Optimal Redundancy]
	The optimal redundancy $ r_{X}^f(\mathcal{X}, k, t) $ is defined as the smallest integer $r$ for which there exists an FCC under the distance $d_X$ with an encoding function $\mathrm{Enc} : \mathcal{X}^k \to \mathcal{X}^{k+r}$ that enables recovery of $f(\mathbf{u})$ from a received vector $\mathbf{y} \in \mathcal{X}^{k+r}$ satisfying
	$d_X(\mathrm{Enc}(\mathbf{u}), \mathbf{y}) \le t $.
\end{defn}
Next, we define the \textit{distance requirement  matrix} associated with a function $f$ as follows.

\begin{defn}[Distance Requirement Matrix under distance $d_X$] 
	Let $\mathbf{u}_1, \dots, \mathbf{u}_M \in \mathcal{X}^{k}$. The distance requirement matrix (DRM) under the distance $d_X$ of a function $f$, $\mathbf{D}_f(t, \mathbf{u}_1, \dots, \mathbf{u}_M)$ , is defined as the $M \times M$ matrix with entries
	$$
	[\mathbf{D}_f(t, \mathbf{u}_1, \dots, \mathbf{u}_M)]_{ij}=
	\begin{cases}
		\begin{aligned}
			&[2t+1-d_{X}(\mathbf{u}_i, \mathbf{u}_j)]^+, \text{if } f(\mathbf{u}_i) \neq f(\mathbf{u}_j), \\
			&0, \quad \text{otherwise}.
		\end{aligned}
	\end{cases}
	$$ 
	
\end{defn}
A \textit{code} $\mathcal{C}$ of length $r$ and size $M$  is a set of $M$ vectors, i.e., $\mathcal{C} =\{\mathbf{u}_i=(u_{i_1},u_{i_2},...,u_{i_r}) : u_{i_j} \in \mathcal{X}, j \in [r], i \in [M] \} $.  Each vector $\mathbf{u}_i\in\mathcal{C}$ is called a \textit{codeword}. Given a matrix $\mathbf{D} \in \mathbb{N}_0^{M \times M}$, we define a $\mathbf{D}_X$-code as follows.  
\begin{defn}[$\mathbf{D}_X$-code]
	\label{Dcode}
	Let $\mathbf{D} \in \mathbb{N}_0^{M \times M}$ and $\mathcal{P} = \{ \mathbf{p}_1, \mathbf{p}_2, \ldots, \mathbf{p}_M \} \subseteq \mathcal{X}^r$ be a code of length $r$ and size $M$. Then, $\mathcal{P} = \{ \mathbf{p}_1, \mathbf{p}_2, \ldots, \mathbf{p}_M \}$ is a $\mathbf{D}_X$-code under the Wyner-Graham distance $d_X$, if there exists an ordering of the codewords of $\mathcal{P}$ such that $d_{X}(\mathbf{p}_i, \mathbf{p}_j) \geq [\mathbf{D}]_{ij}$ for all $i, j \in [M]$.
\end{defn}
The following example illustrates the defintions of DRM and $\mathbf{D}_X$-code.
\begin{example}
Consider the message space $\mathbb{F}_2^3$ equipped with the Hamming distance. Let $f:\mathbb{F}_2^3 \rightarrow {0,1,2}$ be defined by $f(000)=0, f(111)=0, f(011)=1,  f(101)=1, f(001)=2, f(010)=2,  f(100)=2,  f(110)=2.$ The DRM corresponding to this function for $t=1$ and the inputs ordered lexicographically is given by
	\[
	\mathbf{D}_f(1,000,...,111)=
	\begin{pmatrix}
		0&2&2&1&2&1&1&0\\
		2&0&0&2&0&2&0&1\\
		2&0&0&2&0&0&0&1\\
		1&2&2&0&0&0&1&2\\
		2&0&0&0&0&2&0&1\\
		1&2&0&0&2&0&1&2\\
		1&0&0&1&0&1&0&2\\
		0&1&1&2&1&2&2&0
	\end{pmatrix}.
	\]
	
	Consider $\mathcal{P}=\{\mathbf{p_1},\mathbf{p_2},\ldots,\mathbf{p_8}\}$, where 	
	\[
	\begin{aligned}
		\mathbf p_1&=0000, &
		\mathbf p_2&=0011, &
		\mathbf p_3&=0101, &
		\mathbf p_4&=0110, \\
		\mathbf p_5&=0111, &
		\mathbf p_6&=0100, &
		\mathbf p_7&=0010, &
		\mathbf p_8&=0001.
	\end{aligned}
	\] 
	It can be verified that $d_H(\mathbf{p}_i,\mathbf{p}_j)\geq [\mathbf{D}_f(1,000,...,111)]_{ij},
	 \forall i,j\in[8]$. Hence, $\mathcal{P}$ is a $\mathbf{D}_H$-code corresponding to the DRM $\mathbf{D}_f(1,000,...,111)$.
\end{example}
 The following lemma will be used in the derivation of the general Plotkin-type bound in Section \ref{sec_Plotkin_WG}.
\begin{lem}\cite{WG}\label{lem:Plotkin}
	Let $\mathcal{P} = \{ \mathbf{p}_1, \mathbf{p}_2, \ldots, \mathbf{p}_M \} \subseteq \mathcal{X}^r$ be a code of length $r$ and size $M$ with respect to the Wyner-Graham distance $d_X$. Then, we have
	\begin{equation}\label{eq:Plotkin}
		\displaystyle \sum_{ i,j \in [M] : i <j} d_X(\mathbf{p}_i,\mathbf{p}_j) \le \frac{S_XM^2r}{2|\mathcal{X}|}.
	\end{equation} 
\end{lem}

\section{General Plotkin-type Bound for FCCs with Wyner-Graham distance }\label{sec_Plotkin_WG}
In this section, we obtain a general Plotkin-type bound on the optimal redundancy of FCCs with the Wyner-Graham distance,  applicable to arbitrary functions. Consider the message space $\mathcal{X}^k = \{\mathbf{u}_i : i \in [|\mathcal{X}|^k] \}$. Any function $f : \mathcal{X}^k \to Im(f)$ partitions the message space $\mathcal{X}^k$ into disjoint preimage sets corresponding to the distinct values of $f$.  The following theorem establishes a general Plotkin-type bound on the optimal redundancy of FCCs with the Wyner-Graham $d_X$, parameterized by the resulting partition of the message space, specifically through the sizes of the partition classes and the pairwise distances among vectors within each class. As a result, it offers a significantly lower computational complexity compared to existing Plotkin-type bounds.
 \begin{thm}\label{thm:general_plotkin}
 	Let $f:\mathcal{X}^k \to \operatorname{Im}(f)$ be a function with $Im(f)=\{\alpha_1,\alpha_2,...,\alpha_E\}$. For each $ i \in [E]$, let $\mathcal{U}_{\alpha_{i}} \triangleq \{\mathbf{u} \in \mathcal{X}^k : f(\mathbf{u})=\alpha_i \}$ and $n_i = |\mathcal{U}_{\alpha_{i}}|$. Then, the optimal redundancy $ r_{X}^f(\mathcal{X}, k, t)$ of an FCC for $f$ satisfies 
 	\begin{equation}
 		r_{X}^f(\mathcal{X}, k, t) \ge \frac{(2t+1)\left(|\mathcal{X}|^{2k}-\displaystyle \sum_{i=1}^{E} {n_i}^2 \right)+ \displaystyle \sum_{i=1}^{E} \sum_{\mathbf{u},\mathbf{v} \in \mathcal{U}_{\alpha_{i}} } d_{X}(\mathbf{u}, \mathbf{v})}{S_X|\mathcal{X}|^{2k-1}} - k.
 	\end{equation}
 \end{thm}
 \begin{IEEEproof}
 	Consider the message set $\mathcal{X}^k = \{\mathbf{u}_i : i \in [|\mathcal{X}|^k] \}$. Let $\mathbf{D}_f(t, \mathbf{u}_1, \dots, \mathbf{u}_{|\mathcal{X}|^k})$ be the DRM under the Wyner-Graham distance $d_X$ of a function $f:\mathcal{X}^k \to \operatorname{Im}(f)$. For notational convenience, we denote this DRM by $\mathbf{D}$ throughout this proof. Let $\mathcal{P} = \{ \mathbf{p}_1, \mathbf{p}_2, \ldots, \mathbf{p}_{|\mathcal{X}|^k} \} \subseteq \mathcal{X}^r$ be a $\mathbf{D}_X$-code corresponding to the DRM $\mathbf{D}$. From Definition \ref{Dcode}, it follows that $$ d_{X}(\mathbf{p}_i, \mathbf{p}_j) \geq [\mathbf{D}]_{ij}, \text{ for all } i, j \in [|\mathcal{X}|^k]. $$  Therefore, we have 
 	\begin{equation}\label{eq:sum}
 		\sum_{i,j \in [|\mathcal{X}|^k] : i <j} [\mathbf{D}]_{ij} \le \sum_{i,j \in [|\mathcal{X}|^k] : i <j} d_X(\mathbf{p}_i, \mathbf{p}_j).
 	\end{equation}
 	Combining (\ref{eq:Plotkin}) in Lemma \ref{lem:Plotkin} and (\ref{eq:sum}), we obtain
 	\begin{equation}\label{eq:Plotkin1}
 		\sum_{i,j \in [|\mathcal{X}|^k] : i <j} [\mathbf{D}]_{ij} \le \frac{S_X|\mathcal{X}|^{2k}r}{2|\mathcal{X}|} = \frac{S_X|\mathcal{X}|^{2k-1}r}{2}. 
 	\end{equation}
 	
 	Next, we find a lower bound on $\displaystyle \sum_{i,j \in [|\mathcal{X}|^k] : i <j} [\mathbf{D}]_{ij}$. Let $Im(f)=\{\alpha_1,\alpha_2,...,\alpha_E\}$. We define the sets $\mathcal{U}_{\alpha_{i}} \triangleq \{\mathbf{u} \in \mathcal{X}^k : f(\mathbf{u})=\alpha_i \}$ and $n_i = |\mathcal{U}_{\alpha_{i}}|$, $\forall i \in [E]$. Clearly, $\displaystyle \sum_{i=1}^{E} n_i = |\mathcal{X}|^k$. Therefore, $\displaystyle \sum_{i,j \in [|\mathcal{X}|^k] : i <j} [\mathbf{D}]_{ij}$ can be written as 
 	$$ \sum_{i,j \in [|\mathcal{X}|^k] : i <j} [\mathbf{D}]_{ij} = \frac{1}{2} \sum_{i=1}^{E} \sum_{\mathbf{u} \in \mathcal{U}_{\alpha_{i}} } \sum_{\mathbf{v} \in \mathcal{X}^k } [\mathbf{D}(\mathbf{u}, \mathbf{v})],$$ where $[\mathbf{D}(\mathbf{u}, \mathbf{v})]$ denotes the entry of $\mathbf{D}$ corresponding to row $\mathbf{u}$ and column $\mathbf{v}$. This expression can be further simplified using the definition of DRM as follows.
 
 	\begin{align} \label{eq:ub}
 		\sum_{i,j \in [|\mathcal{X}|^k] : i <j} [\mathbf{D}]_{ij} & = \frac{1}{2} \sum_{i=1}^{E} \sum_{\mathbf{u} \in \mathcal{U}_{\alpha_{i}} } \sum_{\mathbf{v} \in \mathcal{X}^k } [D(\mathbf{u}, \mathbf{v})] \nonumber \\ & = \frac{1}{2} \sum_{i=1}^{E} \sum_{\mathbf{u} \in \mathcal{U}_{\alpha_{i}} } \sum_{\mathbf{v} \in \mathcal{X}^k \backslash \mathcal{U}_{\alpha_{i}} } [2t+1-d_{X}(\mathbf{u}, \mathbf{v})]^+ \nonumber \\ & \ge \frac{1}{2} \sum_{i=1}^{E} \sum_{\mathbf{u} \in \mathcal{U}_{\alpha_{i}} } \sum_{\mathbf{v} \in \mathcal{X}^k \backslash \mathcal{U}_{\alpha_{i}} } ( 2t+1-d_{X}(\mathbf{u}, \mathbf{v}) ) \nonumber \\ & = \frac{2t+1}{2} \sum_{i=1}^{E} \left(\sum_{\mathbf{u} \in \mathcal{U}_{\alpha_{i}} } \sum_{\mathbf{v} \in \mathcal{X}^k \backslash \mathcal{U}_{\alpha_{i}} } 1 \right) - \frac{1}{2} \sum_{i=1}^{E} \sum_{\mathbf{u} \in \mathcal{U}_{\alpha_{i}} } \sum_{\mathbf{v} \in \mathcal{X}^k \backslash \mathcal{U}_{\alpha_{i}} } d_{X}(\mathbf{u}, \mathbf{v}).  
 	\end{align}

 	Since $\displaystyle \sum_{\mathbf{u} \in \mathcal{U}_{\alpha_{i}} } \sum_{\mathbf{v} \in \mathcal{X}^k \backslash \mathcal{U}_{\alpha_{i}} } 1 = n_i (|\mathcal{X}|^k-n_i)$ and $\displaystyle \sum_{i=1}^{E} n_i =|\mathcal{X}|^k$, (\ref{eq:ub}) can be written as 
 	\begin{align}\label{eq:inter_class}
 		\sum_{i,j \in [|\mathcal{X}|^k] : i <j} [\mathbf{D}]_{ij} & \ge \frac{2t+1}{2} \sum_{i=1}^{E} \left(n_i (|\mathcal{X}|^k-n_i) \right) - \frac{1}{2} \sum_{i=1}^{E} \sum_{\mathbf{u} \in \mathcal{U}_{\alpha_{i}} } \sum_{\mathbf{v} \in \mathcal{X}^k \backslash \mathcal{U}_{\alpha_{i}} } d_{X}(\mathbf{u}, \mathbf{v}) \nonumber \\ & = \frac{2t+1}{2}  \left(|\mathcal{X}|^{2k}-\sum_{i=1}^{E} {n_i}^2 \right) - \frac{1}{2} \sum_{i=1}^{E} \sum_{\mathbf{u} \in \mathcal{U}_{\alpha_{i}} } \sum_{\mathbf{v} \in \mathcal{X}^k \backslash \mathcal{U}_{\alpha_{i}} } d_{X}(\mathbf{u}, \mathbf{v}) \nonumber \\ & = \frac{2t+1}{2}  \left(|\mathcal{X}|^{2k}-\sum_{i=1}^{E} {n_i}^2 \right) - \frac{1}{2}  \left(\sum_{\mathbf{u}, \mathbf{v} \in \mathcal{X}^k} d_{X}(\mathbf{u}, \mathbf{v}) - \sum_{i=1}^{E} \sum_{\mathbf{u},\mathbf{v} \in \mathcal{U}_{\alpha_{i}} } d_{X}(\mathbf{u}, \mathbf{v}) \right).
 	\end{align}
 	From Lemma \ref{lem_WG}, we have  
 	\begin{equation}\label{eq:xy}
 		\sum_{\mathbf{u}, \mathbf{v} \in \mathcal{X}^k} d_{X}(\mathbf{u}, \mathbf{v}) = k|\mathcal{X}|^{2k-1}S_X.
 	\end{equation}
 	Substituting (\ref{eq:xy}) in (\ref{eq:inter_class}), we obtain
 	\begin{equation}\label{eq:Plotkin2}
 	\sum_{i,j \in [|\mathcal{X}|^k] : i <j} [\mathbf{D}]_{ij} \ge \frac{2t+1}{2}  \left(|\mathcal{X}|^{2k}-\sum_{i=1}^{E} {n_i}^2 \right) - \frac{k|\mathcal{X}|^{2k-1}S_X}{2}   + \frac{1}{2} \sum_{i=1}^{E} \sum_{\mathbf{u},\mathbf{v} \in \mathcal{U}_{\alpha_{i}} } d_{X}(\mathbf{u}, \mathbf{v}). 
 	\end{equation}
 	Therefore, from (\ref{eq:Plotkin1}) and (\ref{eq:Plotkin2}), we get 
 	\begin{equation}\label{eq:Plotkin3}
 	\frac{2t+1}{2}  \left(|\mathcal{X}|^{2k}-\sum_{i=1}^{E} {n_i}^2 \right) - \frac{k|\mathcal{X}|^{2k-1}S_X}{2}   + \frac{1}{2} \sum_{i=1}^{E} \sum_{\mathbf{u},\mathbf{v} \in \mathcal{U}_{\alpha_{i}} } d_{X}(\mathbf{u}, \mathbf{v}) \le \frac{S_X|\mathcal{X}|^{2k-1}r}{2}.
 	\end{equation}
 	Rearranging (\ref{eq:Plotkin3}), we obtain
 	\begin{equation}\label{eq:Plotkin4}
 		r \ge \frac{(2t+1)\left(|\mathcal{X}|^{2k}-\displaystyle \sum_{i=1}^{E} {n_i}^2 \right)+ \displaystyle \sum_{i=1}^{E} \sum_{\mathbf{u},\mathbf{v} \in \mathcal{U}_{\alpha_{i}} } d_{X}(\mathbf{u}, \mathbf{v})}{S_X|\mathcal{X}|^{2k-1}} - k.
 	\end{equation}
 \end{IEEEproof}

 As stated before, the Wyner-Graham distance includes the Hamming distance and the Lee distance. Next, we specialize the obtained bound in Theorem \ref{thm:general_plotkin} to the Hammig distance.  Let $\mathcal{X}=\mathbb{F}_q$. Then $|\mathcal{X}|=q$ and for Hamming distance, we have $S_H=q-1$. Substituting these quantities into the general Plotkin-type bound yields the following corollary for FCCs with the Hamming distance. The optimal redundancy $ r_{X}^f(\mathcal{X}, k, t) $ gets  denoted by $ r_{H}^f(q, k, t) $. 
 \begin{cor}\label{cor:Hamming_plotkin}
 	Let $f:\mathbb{F}_q^k \to \operatorname{Im}(f)$ be a function with $Im(f)=\{\alpha_1,\alpha_2,...,\alpha_E\}$. For each $ i \in [E]$, let $\mathcal{U}_{\alpha_{i}} \triangleq \{\mathbf{u} \in \mathbb{F}_q^k : f(\mathbf{u})=\alpha_i \}$ and $n_i = |\mathcal{U}_{\alpha_{i}}|$. Then, the optimal redundancy $ r_{H}^f(q, k, t)$ of an FCC for $f$ satisfies 
 	\begin{equation}\label{eq:Hamming_plotkin}
 		r_{H}^f(q, k, t) \ge \frac{(2t+1)\left(q^{2k}-\displaystyle \sum_{i=1}^{E} {n_i}^2 \right)+ \displaystyle \sum_{i=1}^{E} \sum_{\mathbf{u},\mathbf{v} \in \mathcal{U}_{\alpha_{i}} } d_{H}(\mathbf{u}, \mathbf{v})}{(q-1)q^{2k-1}} - k.
 	\end{equation}
 \end{cor}
 The known Plotkin-type bound on the optimal redundancy of FCCs with the Hamming distance, applicable to arbitrary function, is stated in the following lemma.
 \begin{lem}\cite{LBWY}\label{lem:Lenz}
 	Let $\mathbf{D}_f(t, \mathbf{u}_1, \dots, \mathbf{u}_{q^k})$ denote the DRM associated with a function $f$ under the Hamming distance. Then 
 	\begin{equation}
 		r_{H}^f(q, k, t) \geq \begin{cases}
 			\frac{4}{q^{2k}} \displaystyle \sum_{i,j:i<j}[\mathbf{D}_f(t, \mathbf{u}_1, \dots, \mathbf{u}_{q^k})]_{ij},  \text{ if } q \text{ is even}, \\ \frac{4}{q^{2k}-1} \displaystyle \sum_{i,j:i<j}[\mathbf{D}_f(t, \mathbf{u}_1, \dots, \mathbf{u}_{q^k})]_{ij},  \text{ if } q \text{ is odd.}
 		\end{cases} 
 	\end{equation}
 \end{lem}
 
 From Lemma \ref{lem:Lenz} and the definition of DRM, it is clear that this bound depends on the distance requirements between pairs of message vectors corresponding to distinct function values. Consequently, evaluating the resulting bounds becomes computationally intensive as the size of the message space grows. However, the proposed bound in Corollary \ref{cor:Hamming_plotkin} requires the computation of the cardinalities of the preimage sets and the corresponding sums of pairwise distances only among vectors belonging to the same preimage set, which is relatively easier.
 
  Next, we specialize the general Plotkin-type bound stated in Theorem \ref{thm:general_plotkin} to the Lee distance. The Lee distance is a special case of   the  Wyner-Graham distance \cite{WG}. Let $\mathcal{X}=\mathbb{Z}_q$, where $q\ge 2$.  Therefore, $|\mathcal{X}|=q$, and we have the expression for $S_L$ in (\ref{eq:lee}). Substituting these quantities into the general Plotkin-type bound yields the following corollary for FCCs with the Lee distance. Moreover, the optimal redundancy $r_X^f(\mathcal{X},k,t)$ is now denoted by $r_L^f(q,k,t)$.
  \begin{cor}\label{cor:Lee_plotkin}
  	Let $f:\mathbb{Z}_q^k \to \operatorname{Im}(f)$ be a function with $Im(f)=\{\alpha_1,\alpha_2,...,\alpha_E\}$. For each $ i \in [E]$, let $\mathcal{U}_{\alpha_{i}} \triangleq \{\mathbf{u} \in \mathbb{Z}_q^k : f(\mathbf{u})=\alpha_i \}$ and $n_i = |\mathcal{U}_{\alpha_{i}}|$. Then, the optimal redundancy $ r_{L}^f(q, k, t)$ of an FCC under the Lee distance for $f$ satisfies 
  	\begin{equation}
  		r_{L}^f(q, k, t) \ge \frac{(2t+1)\left(q^{2k}-\displaystyle \sum_{i=1}^{E} {n_i}^2 \right)+ \displaystyle \sum_{i=1}^{E} \sum_{\mathbf{u},\mathbf{v} \in \mathcal{U}_{\alpha_{i}} } d_{L}(\mathbf{u}, \mathbf{v})}{S_Lq^{2k-1}} - k, 
  	\end{equation}
  	where
  	\begin{equation*}
  		S_L=
  		\begin{cases}
  			\frac{q^2}{4}, & q \text{ is even},\\
  			\frac{q^2-1}{4}, & q \text{ is odd}.
  		\end{cases}
  	\end{equation*}
  \end{cor}
  
  A Plotkin-type bound on the optimal redundancy of FCCs under the Lee distance, applicable to arbitrary functions, was proposed in \cite{HRS} and is stated in the following lemma.
  \begin{lem}\cite{HRS}\label{lem:Lee}
  	Let $\mathbf{D}_f(t, \mathbf{u}_1, \dots, \mathbf{u}_{q^k})$ denote the DRM associated with a function $f$ under the Lee distance. Then 
  	\begin{equation}
  		r_{L}^f(q, k, t) \geq \begin{cases}
  			\begin{aligned}
  				&\frac{8}{q^{2k+1}}\sum_{i,j: i<j} [\mathbf{D}_f(t, \mathbf{u}_1, \dots, \mathbf{u}_{q^k})]_{ij},  \text{ if } q \text{ is even,} \\
  				&\frac{8q}{q^{2k}(q^2-1)}\sum_{i,j: i<j} [\mathbf{D}_f(t, \mathbf{u}_1, \dots, \mathbf{u}_{q^k})]_{ij},  \text{ if } q \text{ is odd}.
  			\end{aligned}
  		\end{cases}
  	\end{equation}
  \end{lem}
  
  From Lemma \ref{lem:Lee}, it is clear that the bound depends on the entries of the DRM, and hence its evaluation becomes computationally intensive as the size of the message space grows. In contrast, as in the Hamming distance case, the proposed bound in Corollary \ref{cor:Lee_plotkin} requires the computation of the cardinalities of the preimage sets and the corresponding sums of pairwise distances only among vectors belonging to the same preimage set, which is relatively easier.
 
\section{Plotkin-type Bound for FCCs for Linear Functions }\label{sec_Plotkin_linear}

  In this section, we simplify  the general Plotkin-type bound stated in Theorem \ref{thm:general_plotkin} for linear functions under the Wyner-Graham distance. This is achieved by explicitly characterizing the induced partition of the message space through the sizes of the partition classes and the pairwise distances among vectors within each class. In this section, we assume that the alphabet $\mathcal{X}$ is equipped with two operations $'+'$ and $'.'$ such that $(\mathcal{X},+,.)$ forms a finite ring.  This additional algebraic structure allows us to consider linear functions $f : \mathcal{X}^{k} \to \mathcal{X}^{l} $. Let the additive identity be $0$. We define the Wyner-Graham weight of a symbol $x \in \mathcal{X}$, denoted by $w_X(x)$, as $w_X(x) = d_X(x,0)$. The Wyner-Graham weight of a vector $\mathbf{x} \in \mathcal{X}^{n}$ is then defined by $w_X(\mathbf{x}) = d_X(\mathbf{x},\mathbf{0})$. The formal definition and properties of linear function is stated below.
  \begin{defn}[Linear Function under the Wyner-Graham distance]
  	Let $\mathcal{X}$ be a finite ring. A function $f : \mathcal{X}^{k} \to \mathcal{X}^{l} $ is said to be linear if it satisfies the following condition: $$f(\alpha . \mathbf{x} + \beta . \mathbf{y}) = \alpha . f(\mathbf{x}) +\beta. f(\mathbf{y}), \forall\hspace{0.1cm} \mathbf{x},\mathbf{y} \in \mathcal{X}^{k} \text{ and } \alpha , \beta \in  \mathcal{X}.$$ The kernel of $f$, denoted by $ker(f)$, is defined as $ker(f) \triangleq \{\mathbf{x} \in \mathcal{X}^{k} : f(\mathbf{x}) =\mathbf{0} \}$, where $\mathbf{0}$ denotes the zero vector in $\mathcal{X}^{l}$. 
  \end{defn} 
  
  The set of all cosets of $ker(f)$ is denoted by $\mathcal{X}^{k}/ker(f)$. Since $f : \mathcal{X}^{k} \to \mathcal{X}^{l} $ is linear, the cosets in $\mathcal{X}^{k}/ker(f)$ form a partition of $\mathcal{X}^{k}$, and each coset consists of all vectors that map to a given function value. Therefore, we have $\mathcal{U}_{\alpha_{i}} = \mathbf{u}_i + ker(f)$ for some $\mathbf{u}_i$ and
  \begin{equation}\label{eq:linear1}
  n_i=|\mathcal{U}_{\alpha_{i}}|={|\mathcal{X}|}^{k-l}, \forall i \in [E].
  \end{equation}  
  When the Wyner-Graham distance is translation-invariant, all cosets have the same distance distribution, which is same as the distance distribution of $ker(f)$. Therefore, we have 
  \begin{equation}\label{eq:linear2}
  	\sum_{\mathbf{u},\mathbf{v} \in \mathcal{U}_{\alpha_{i}} } d_{X}(\mathbf{u}, \mathbf{v}) = \sum_{\mathbf{u},\mathbf{v} \in ker(f) } d_{X}(\mathbf{u}, \mathbf{v}), \forall i \in [E].
  \end{equation}
  
  A distance $d_X: \mathcal{X} \times \mathcal{X} \to \mathbb{N}_0$ is said to be \textit{translation-invariant} if $d_X(x+z,y+z)=d_X(x,y), \forall x,y,z \in \mathcal{X}$. Equivalently, we have $d_X(x,y)=d_X(x-y,0)=w_X(x-y)$ and $d_X(\mathbf{x},\mathbf{y})=w_X(\mathbf{x}-\mathbf{y})$.  When the Wyner-Graham distance is  translation-invariant, we have $\displaystyle  \sum_{\mathbf{u},\mathbf{v} \in ker(f) } d_{X}(\mathbf{u}, \mathbf{v})= \displaystyle \sum_{\mathbf{u},\mathbf{v} \in ker(f) } w_{X}(\mathbf{u}-\mathbf{v})$. Since $ker(f)$ is linear, for every fixed $\mathbf{v} \in ker(f)$, we have $\{\mathbf{u}-\mathbf{v} : \mathbf{u} \in ker(f)\} =ker(f)$. Let $s$ denotes the sum of Wyner-Graham weights of the vectors in $ker(f)$,  i.e, $s = \displaystyle \sum_{\mathbf{w}\in ker(f)}w_{X}(\mathbf{w}) $. Therefore, we obtain 
  \begin{equation}\label{eq:linear3}  
  	\sum_{\mathbf{u},\mathbf{v} \in ker(f) } d_{X}(\mathbf{u}, \mathbf{v})= \displaystyle \sum_{\mathbf{u},\mathbf{v} \in ker(f) } w_{X}(\mathbf{u}-\mathbf{v}) =|ker(f)| \displaystyle \sum_{\mathbf{w}\in ker(f)}w_{X}(\mathbf{w}) = {|\mathcal{X}|}^{k-l} s. 
  \end{equation}

Substituting $E=|Im(f)|=|\mathcal{X}|^l$, (\ref{eq:linear1}) and (\ref{eq:linear2}) in Theorem \ref{thm:general_plotkin}, we obtain a simplified lower bound on the optimal redundancy of FCCs for linear functions as follows.  
\allowdisplaybreaks
\begin{align}
	r_{X}^f(\mathcal{X}, k, t) & \ge  \frac{(2t+1)\left(|\mathcal{X}|^{2k}-|\mathcal{X}|^l {|\mathcal{X}|}^{2k-2l} \right)+ |\mathcal{X}|^l \displaystyle \sum_{\mathbf{u},\mathbf{v} \in ker(f) } d_{X}(\mathbf{u}, \mathbf{v})}{S_X|\mathcal{X}|^{2k-1}} - k \nonumber \\ & = \frac{(2t+1)|\mathcal{X}|(1-{|\mathcal{X}|}^{-l})}{S_X} +\frac{|\mathcal{X}|^l \displaystyle \sum_{\mathbf{u},\mathbf{v} \in ker(f) } d_{X}(\mathbf{u}, \mathbf{v})}{S_X|\mathcal{X}|^{2k-1}}-k. \nonumber
\end{align} 
The obtained lower bound on the optimal redundancy of FCCs for linear functions and its simplified form, obtained using (\ref{eq:linear3}), when the Wyner Graham-distance is translation-invariant, are stated in the following corollary.  
  \begin{cor}\label{cor_linear}
  	Let $f : \mathcal{X}^{k} \to \mathcal{X}^{l} $ be a linear function with $ker(f) \triangleq \{\mathbf{u} \in \mathcal{X}^{k} : f(\mathbf{u}) =\mathbf{0} \}$. Then, the optimal redundancy $ r_{X}^f(\mathcal{X}, k, t)$ of an FCC for $f$ with Wyner Graham distance satisfies
  	\begin{equation}
  	r_{X}^f(\mathcal{X}, k, t) \ge  \frac{(2t+1)|\mathcal{X}|(1-{|\mathcal{X}|}^{-l})}{S_X} +\frac{|\mathcal{X}|^l \displaystyle \sum_{\mathbf{u},\mathbf{v} \in ker(f) } d_{X}(\mathbf{u}, \mathbf{v})}{S_X|\mathcal{X}|^{2k-1}}-k. 
  	\end{equation} 
  	Furthermore, when the Wyner-Graham distance is translation-invariant, the bound simplies to
  	\begin{equation}\label{eq_linear}
  		r_{X}^f(\mathcal{X}, k, t) \ge \frac{(2t+1)|\mathcal{X}|(1-{|\mathcal{X}|}^{-l})}{S_X} +\frac{s}{S_X|\mathcal{X}|^{k-1}}-k,
  	\end{equation}
  	where $s = \displaystyle \sum_{\mathbf{w}\in ker(f)}w_{X}(\mathbf{w}) $.
  \end{cor}
  
  We now specialize the above result to the Hamming distance, which is both translation-invariant and a special case of the Wyner-Graham distance. We have, $\mathcal{X}=\mathbb{F}_q$, $S_X=S_H=q-1$ and $s = \displaystyle \sum_{\mathbf{w}\in ker(f)}w_{X}(\mathbf{w}) $. Substituting these quantities into the bound in Corollary \ref{cor_linear}, a lower bound on optimal redundancy of FCCs for linear functions with Hamming distance is obtained as 
  \begin{equation}\label{bound_linear}
  		r_{H}^f(q, k, t)  \ge  \frac{q(2t+1)(1-q^{-l})}{q-1} +\frac{s}{(q-1)q^{k-1}}-k.
  \end{equation}
  The bound in (\ref{bound_linear}) is exactly the same as the Plotkin bound for FCCs for linear functions under the Hamming distance obtained in \cite{PR}. 
  
  We now specialize the result to the Lee distance. The Wyner-Graham distance includes the Lee distance as special case, and the Lee distnce is also translation-invariant. For Lee distance, we have $\mathcal{X}=\mathbb{Z}_q$, $S_X=S_L$ is given by (\ref{eq:lee}) and $s = \displaystyle \sum_{\mathbf{w}\in ker(f)}w_{L}(\mathbf{w}) $. Substituting these quantities in (\ref{eq_linear}), for a linear function $f:\mathbb{Z}_{q}^{k} \to \mathbb{Z}_{q}^{l}$, the optimal redundancy of an FCC under the Lee distance satisfies
  \begin{equation}\label{bound_linear_Lee}
  	r_{L}^f(q, k, t)  \ge \frac{q(2t+1)(1-q^{-l})}{S_L} +\frac{s}{S_Lq^{k-1}}-k.
  \end{equation} 
  The bound in (\ref{bound_linear_Lee}) is exactly the same as the Plotkin bound for FCCs under the Lee distance for linear functions obtained in \cite{GA}.  
\section{Simplified Plotkin-type Bounds for functions under the Hamming distance }\label{sec_Plotkin_Ham} 

In this section, we simplify the proposed bound for FCCs under the Hamming distance in Corollary \ref{cor:Hamming_plotkin} to several classes of functions, including the Hamming weight function, the Hamming weight distribution function and the monomial function. This is achieved by explicitly characterizing the induced partition of the message space through the sizes of the partition classes and the pairwise distances among vectors within each class.

\subsection{Hamming weight function}
Consider the Hamming weight function  $f(\mathbf{u}) = {w}_H(\mathbf{u})=\displaystyle \sum_{i=1}^k 1_{(u_i \ne 0)}$, where $\mathbf{u} \in \mathbb{F}_q^k$, and \\ $1_{(u_i \ne 0)}= \begin{cases}
	1, & u_i \ne 0, \\
	0, & u_i =0.
\end{cases}$ The expressiveness of this function is $E=k+1$, and $Im(f)=\{0,1,...,k\}$.  For each $ i \in [0:k]$, we have $\mathcal{U}_{i} \triangleq \{\mathbf{u} \in \mathbb{F}_q^k : {w}_H(\mathbf{u})=i \}$.  Every vector $\mathbf{u} \in \mathcal{U}_{i}$ has exactly $i$ non zero coordinates, each of which can take any value in $\mathbb{F}_q \textbackslash \{0\}$.  Therefore, $n_i = |\mathcal{U}_{i}|=\binom{k}{i}(q-1)^i$. 
Therefore, we have
\begin{equation}\label{eq:sum_squares_q}
	 \sum_{i=0}^{k} n_i^2 = \sum_{i=0}^{k} \binom{k}{i}^{2}(q-1)^{2i}.
\end{equation}

 Now we proceed to show that $\displaystyle \sum_{i=1}^{E} \sum_{\mathbf{u},\mathbf{v} \in \mathcal{U}_{\alpha_{i}} } d_{H}(\mathbf{u}, \mathbf{v}) = \sum_{i=0}^{k} {n_i}^2 \left( 2i-\frac{qi^2}{k(q-1)} \right)$.
 
Let $S_i=\displaystyle \sum_{\mathbf{u},\mathbf{v} \in \mathcal{U}_{i} } d_{H}(\mathbf{u}, \mathbf{v})$. Using  $d_H(\mathbf{u},\mathbf{v})=\displaystyle \sum_{j=1}^{k} 1_{(u_j\neq v_j)}$, we obtain $S_i=\displaystyle \sum_{j=1}^{k}\sum_{\mathbf{u},\mathbf{v}\in\mathcal{U}_i} 1_{(u_j\neq v_j)}$. Let $N_{i,j} = \displaystyle \sum_{\mathbf{u},\mathbf{v}\in\mathcal{U}_i} 1_{(u_j\neq v_j)}$ for $1\le j\le k$, which counts the number of ordered pairs in $ |\mathcal U_i|^2$ that differ in the $j$-th coordinate. A set $\mathcal U\subseteq \mathcal{X}^k$ is said to be \emph{invariant under coordinate permutations} if, for every	$\mathbf u=(u_1,\ldots,u_k)\in\mathcal{U}$ and every permutation $\pi(l)$, where $l \in [k]$,  $(u_{\pi(1)},u_{\pi(2)},\ldots,u_{\pi(k)})\in\mathcal{U}$. Since $\mathcal{U}_i$ is invariant under coordinate permutations, the quantities $N_{i,j}$ are the same for all $j$, allowing us to denote this common value simply as $N_i$. Therefore, $S_i = \displaystyle \sum_{\mathbf{u},\mathbf{v} \in \mathcal{U}_{i} } d_{H}(\mathbf{u}, \mathbf{v}) = \sum_{j=1}^{k} N_{i,j} = kN_i$.

To determine $N_i$, let $B_a=\bigl|\{\mathbf{u}\in\mathcal{U}_i:u_i=a\}\bigr|$ for $a\in\mathbb F_q$ and any fixed coordinate $j$. Then the number of ordered pairs that agree in this coordinate is $\sum_{a\in\mathbb F_q} B_a^2$. Since there are $n_i^2$ ordered pairs in total, it follows that $N_i=n_i^2-\sum_{a\in\mathbb F_q} B_a^2$.

We evaluate these cardinalities by splitting $a$ into the zero and non-zero cases:
\begin{align*}
	B_0 &= \binom{k-1}{i}(q-1)^i, \quad \text{for } a = 0, \\
	B_a &= \binom{k-1}{i-1}(q-1)^{i-1}, \quad \text{for } a \in \mathbb{F}_q^{\ast}.
\end{align*}
Since there are exactly $q-1$ elements $a \in \mathbb{F}_q^{\ast}$, separating the $a=0$ term from the rest of the summation yields
\[
\sum_{a\in\mathbb F_q} B_a^2 = \binom{k-1}{i}^{2}(q-1)^{2i} + (q-1)\binom{k-1}{i-1}^{2}(q-1)^{2i-2}.
\]
Substituting this summation into the expression $N_i = n_i^2 - \sum_{a\in\mathbb F_q} B_a^2$ and using $S_i = k N_i$, the total distance sum for weight class $i$ expands and simplifies as follows
\begin{align*}
	S_i
	&= k\left[ n_i^2 - B_0^2 - (q-1)B_a^2 \right] 
	= k\left[ n_i^2 -\frac{(k-i)^2}{k^2}n_i^2 -\frac{i^2}{k^2(q-1)}n_i^2 \right] \\
	&= n_i^2 \left[ k-\frac{k^2-2ki+i^2}{k} -\frac{i^2}{k(q-1)} \right] 
	= n_i^2 \left[ 2i-\frac{q i^2}{k(q-1)} \right].
\end{align*}

Summing the contributions $S_i$ across all possible weight classes from $0$ to $k$, we obtain 
\begin{align}\label{eq:s}
	 \sum_{i=0}^{k} \sum_{\mathbf{u},\mathbf{v} \in \mathcal{U}_{\alpha_{i}} } d_{H}(\mathbf{u}, \mathbf{v}) = \sum_{i=0}^{k} S_i 
	= \sum_{i=0}^{k} n_i^2 \left( 2i-\frac{q i^2}{k(q-1)} \right).
\end{align}
Substituting \eqref{eq:sum_squares_q} and \eqref{eq:s} into \eqref{eq:Hamming_plotkin} and simplifying, we obtain
\begin{align}
	r_H^{\mathrm{w}}(q,k,t) & \ge \frac{1}{(q-1)q^{2k-1}} \left[ (2t+1) \left(q^{2k}-\displaystyle \sum_{i=0}^{k} n_i^2 \right)+ \displaystyle \sum_{i=0}^{k} n_i^2 \left( 2i-\frac{q i^2}{k(q-1)} \right) \right] - k \nonumber \\ 
	& = \frac{1}{(q-1)q^{2k-1}} \left[ (2t+1)q^{2k} 
	+ \sum_{i=0}^{k} n_i^2 \left( 2i-(2t+1) -\frac{q i^2}{k(q-1)} \right) \right] - k \nonumber \\
	& = \frac{1}{(q-1)q^{2k-1}} \left[ (2t+1)q^{2k} 
	+ \sum_{i=0}^{k} \binom{k}{i}^{2}(q-1)^{2i} \left( 2i-(2t+1) -\frac{q i^2}{k(q-1)} \right) \right] - k. \label{eq:Ham_wt} 
\end{align}
For  $q=2$, \eqref{eq:Ham_wt} simplifies to
\begin{equation}\label{eq:Ham_wt2}
	r_H^{\mathrm{w}}(k,t)
	\ge 
	\frac{1}{2^{2k-1}} \left[ (2t+1)4^k 
	+  \sum_{i=0}^{k} \binom{k}{i}^2 \left( 2i-\frac{2i^2}{k}-(2t+1) \right) \right] - k.
\end{equation}
The obtained bound for the Hamming weight function is stated in the following corollary.
\begin{cor}\label{cor:plotkin_Ham_weight}
	The optimal redundancy $r_H^{\mathrm{w}}(q,k,t)$ for the Hamming weight function satisfies
	\begin{equation*}
		\label{eq:plotkin_Ham_weight}
		r_H^{\mathrm{w}}(q,k,t)
		\ge 
		\frac{1}{(q-1)q^{2k-1}} \left[ (2t+1)q^{2k} 
		+ \sum_{i=0}^{k} \binom{k}{i}^{2}(q-1)^{2i} \left( 2i-(2t+1) -\frac{q i^2}{k(q-1)} \right) \right] - k.
	\end{equation*}
\end{cor}
\begin{example}
	Consider the Hamming weight function $f(\mathbf{u}) = {w}_H(\mathbf{u})$ for $\mathbf{u} \in\mathbb F_2^3$. For $t=1$, substituting $k=3$ in \eqref{eq:Ham_wt2}, we get  $r_H^{\mathrm{w}}(3,1) \ge \frac{15}{8}$. Since the redundancy must be an integer, we have $r_H^{\mathrm{w}}(3,1) \ge 2$.
\end{example}
Simplified lower bounds on the optimal redundancy of FCCs for the Hamming weight function were derived in \cite{LBWY} and \cite{ZXZG}, but only for $k > t$. Whereas, the obtained lower bound in Corollary \ref{cor:plotkin_Ham_weight} is applicable in all parameter regimes.
\subsection{Hamming weight distribution function}
We now obtain a simplified bound for the Hamming-weight distribution function, defined as $f(\mathbf{u}) = \triangle_T(\mathbf{u}) \triangleq \left\lfloor \frac{{w}_H(\mathbf{u})}{T} \right\rfloor$, where $\mathbf{u} \in \mathbb{F}_q^k$ and $T \in \mathbb{N}$ is chosen such that $T \mid (k+1)$. Therefore, $E = \frac{k+1}{T}$ and $\operatorname{Im}(f) = \{0, 1, \ldots, E-1\}$. For each $j \in \{0, 1, \ldots, E-1\}$, the corresponding preimage set is $\mathcal{V}_j = \left\{ \mathbf{u} \in \mathbb{F}_q^k : jT \le {w}_H(\mathbf{u}) \le (j+1)T-1 \right\}$. Thus, $\mathcal{V}_j$ contains all vectors whose Hamming weights lie in the interval $[jT, (j+1)T-1]$. Consequently, the cardinality of $\mathcal{V}_j$ is 
\begin{equation*}
	N_j = |\mathcal{V}_j| = \sum_{i=jT}^{(j+1)T-1} \binom{k}{i}(q-1)^i. 
\end{equation*}
To apply in Corollary \ref{cor:Hamming_plotkin}, we compute the sum of the squared preimage sizes
\begin{equation}\label{eq:Ham_dis}
	\sum_{j=0}^{\frac{k+1}{T}-1}N_j^2 = \sum_{j=0}^{\frac{k+1}{T}-1} \left( \sum_{i=jT}^{(j+1)T-1} \binom{k}{i}(q-1)^i \right)^2. 
\end{equation} 

Since
$\mathcal{V}_j = \bigcup_{i=jT}^{(j+1)T-1} \mathcal{U}_i $, where $\mathcal{U}_i= \{\mathbf{u}\in\mathbb{F}_q^k: \mathrm{w}_H(\mathbf{u})=i\}$
and the sets $\{\mathcal{U}_i\}_{i=0}^{k}$ are invariant under coordinate permutations and pairwise disjoint, $\mathcal{V}_j$ is also invariant under coordinate permutations. Hence,
$\displaystyle \sum_{\mathbf{u},\mathbf{v}\in\mathcal{V}_j}d_H(u_l,v_l)=\sum_{\mathbf{u},\mathbf{v}\in\mathcal{V}_j}d_H(u_1,v_1),$ for all $l\in[k]$, and therefore $\displaystyle \sum_{\mathbf{u},\mathbf{v}\in\mathcal{V}_j}
d_H(\mathbf{u},\mathbf{v})=k \sum_{\mathbf{u},\mathbf{v}\in\mathcal{V}_j} d_H(u_1,v_1).$ Summing over all $j$ establishes 
\begin{equation}
	\label{eq:Ham_dis2}
	\sum_{j=0}^{\frac{k+1}{T}-1}
	\sum_{\mathbf{u},\mathbf{v}\in\mathcal{V}_j}
	d_H(\mathbf{u},\mathbf{v})
	=
	k
	\sum_{j=0}^{\frac{k+1}{T}-1}
	\sum_{\mathbf{u},\mathbf{v}\in\mathcal{V}_j}
	d_H(u_1,v_1).
\end{equation}
 Substituting \eqref{eq:Ham_dis} and \eqref{eq:Ham_dis2} into \eqref{eq:Hamming_plotkin}, we obain a lower bound on the optimal redundancy of FCC for the Hamming weight distribution function, as stated in the following corollary.
\begin{cor}\label{cor:plotkin_Ham_distribution}
	The optimal redundancy $r_H^{\Delta_T}(q,k,t)$ for the Hamming weight distribution function satisfies
	\begin{align}
		\label{eq:plotkin_Ham_distribution}
		r_L^{\triangle_T}(q,k,t)
		\ge &
		\frac{1}{q^{2k-1}(q-1)}
		\left[
		(2t+1)
		\left(
		q^{2k}
		-\sum_{j=0}^{\frac{k+1}{T}-1} \left( \sum_{i=jT}^{(j+1)T-1} \binom{k}{i}(q-1)^i \right)^2. 
		\right)
		+ \right. \nonumber \\ & \left. \hspace{7cm}
		k
		\sum_{j=0}^{\frac{k+1}{T}-1}
		\sum_{\mathbf{u},\mathbf{v}\in\mathcal{V}_j}
		\!\!d_H(u_1,v_1)
		\right]
		-k.
	\end{align}
\end{cor}

\subsection{Monomial function}

Many-to-one mappings have wide applications in several areas, particularly in cryptography, finite geometry, coding theory, and combinatorial design theory \cite{ZDZYW, MQ, NLQL}. Monomial mappings induce structured many-to-one mappings that exhibit rich combinatorial and algebraic properties, making them particularly relevant for coding-theoretic analysis. Recent studies on many-to-one mappings over finite fields \cite{ZDZYW}, two-to-one mappings \cite{MQ}, and characterisations and constructions of $n$-to-$1$ mappings \cite{NLQL} further highlight the importance of algebraically structured mappings in finite field theory and coding theory, motivating their investigation in the context of FCCs. We first review the definitions of $m$-to-$1$ mapping and monomial functions.

\begin{defn}[$m$-to-$1$ mapping \cite{ZDZYW}]
	Let $\mathcal{A}$ be a finite set and  $m \in \mathbb{Z}$ with
	$1 \leq m \leq |\mathcal{A}|$. Write $|\mathcal{A}| = km + r,$ where $k,r \in \mathbb{Z}$ and $0 \leq r < m$. Let $f: \mathcal{A} \to \mathcal{B}$ be a mapping from $\mathcal{A}$ to another finite set $\mathcal{B}$. Then $f$ is called an $m$-to-$1$ mapping on $\mathcal{A}$ if there exist $k$ distinct elements in $\mathcal{B}$ such that each of these elements has exactly $m$ preimages in $\mathcal{A}$ under $f$.
	The remaining $r$ elements in $\mathcal{A}$ are called the exceptional elements of $f$ on $\mathcal{A}$, and the set of these $r$ exceptional elements is called the exceptional set of $f$ on $\mathcal{A}$, denoted by $E_f(\mathcal{A})$. In particular, $E_f(\mathcal{A})=\emptyset$ if and only if $r=0,$ that is, $m \mid |\mathcal{A}|.$ A polynomial $f(x) \in \mathbb{F}_q[x]$ is called \emph{$m$-to-1} over $\mathbb{F}_q$ if the induced mapping $f : \mathbb{F}_q \to \mathbb{F}_q$, defined by $c \mapsto f(c)$, is $m$-to-1 on $\mathbb{F}_q$.
\end{defn}

\begin{example}
	Let $f(x) = x^3 + x$. Then $f$ maps $0,1,2,3,4$ to $0,2,0,0,3$ in $\mathbb{F}_5$, respectively. Thus, $f$ is $3$-to-$1$ on $\mathbb{F}_5$, and the exceptional set $E_f(\mathbb{F}_5) = \{1,4\}$.   
\end{example}

\begin{defn}[monomial function \cite{ZDZYW}]
	The monomial function $f(x)=x^n$ with $n \in \mathbb{N}$ is $gcd(n, q-1)$-to-$1$ on $\mathbb{F}_q^{*}$, and $E_f(\mathbb{F}_q^{*}) = \emptyset$.
\end{defn}

Next, we obtain a simplified Plotkin-type bound on the optimal redundancy of FCCs for the monomial functions. We also characterise the conditions under which the monomial functions are bijective. We consider monomial functions of the form $f(x)=x^n$ over $\mathbb{F}_{q^k}$ and rewrite the definition of monomial function as the following lemma.
\begin{lem}
	Let $f : \mathbb{F}_{q^k} \to \mathrm{Im}(f)$ be the monomial function defined by $f(x)=x^n,$ where $n \in \mathbb{N}$, and let $\ell=\gcd(n,q^k-1).$ Then every nonzero element in $\mathrm{Im}(f)$ has exactly $\ell$ preimages in $\mathbb{F}_{q^k}^{\ast}$ and $|\mathrm{Im}(f)| = \frac{q^k-1}{\ell}+1.$
\end{lem}

The following lemma shows that the nonempty preimage sets of a monomial function are precisely the multiplicative cosets of its kernel.
\begin{lem}
	Let $f:\mathbb{F}_{q^k}\to \mathrm{Im}(f)$ where $f(x)=x^n$, and define the kernel as $K=\ker(f)=\{x\in \mathbb{F}_{q^k}^{\ast}:x^n=1\}$. Then $K$ is a multiplicative subgroup of $\mathbb{F}_{q^k}^{\ast}$, and each nonempty preimage set of $f$ is a multiplicative coset of $K$.
\end{lem}
\begin{IEEEproof}
	Let
	$K=\ker(f)=\{x\in\mathbb{F}_{q^k}^{\ast}:x^n=1\}$.
	For any $a,b\in K$, we have $a^n=b^n=1$. Hence, $(ab^{-1})^n=a^n(b^{-1})^n=1,$
	which implies $ab^{-1}\in K$ and hence $K$ is a subgroup of $\mathbb{F}_{q^k}^{\ast}$.
	
	Now let $y\in \operatorname{Im}(f)$ and choose $x_0\in\mathbb{F}_{q^k}^{\ast}$ such that
	$x_0^n=y$.
	If $x\in f^{-1}(y)$, then
	$x^n=x_0^n,$
	and therefore
	$(xx_0^{-1})^n=1.$
	Thus,
	$xx_0^{-1}\in K$,
	which implies that there exists $k\in K$ such that
	$x=x_0k$.
	Conversely, for any $k\in K$,
	$(x_0k)^n=x_0^nk^n=y.$
	Hence,
	$f^{-1}(y)=x_0K$.
	Therefore, each nonempty preimage set of $f$ is a multiplicative coset of $K$.
\end{IEEEproof}
From the definition of monomial functions and above lemma, we obtain the following corollary. 
\begin{cor}\label{cor_monomial}
	Let $f \colon \mathbb{F}_{q^k} \to \operatorname{Im}(f)$ be defined by $f(x)=x^n$, and let $\ell=\gcd(n,q^k-1)$ and $M=\frac{q^k-1}{\ell}$. If $\beta$ is a primitive element of $\mathbb{F}_{q^k}$, then the preimage sets of $f$ on $\mathbb{F}_{q^k}^{\ast}$ are given by 
	$$\mathcal{U}_{\beta^g} = \{\beta^{g+pM} : 0 \le p \le \ell-1\}, ~~~~~  0 \le g \le M-1.$$ Equivalently, $\mathcal{U}_{\beta^g} = \left\{ \beta^g, \beta^{g+M}, \beta^{g+2M}, \dots, \beta^{g+(\ell-1)M} \right\}.$
	Moreover, $\mathcal{U}_0=\{0\}.$ 
\end{cor}
 The following two examples illustrate Corollary \ref{cor_monomial}.
\begin{example}
	\label{exm2}
	Consider the monomial mapping $f(x)=x^n$ over $\mathbb{F}_{4^2}$, where $|\mathbb{F}_{4^2}^{*}|=15$ and $\ell=\gcd(n,15)$. Let $\beta$ be a primitive element of $\mathbb{F}_{4^2}$, and represent each field element as $(a,b)\in\mathbb{F}_4^2$.
	
	\textbf{Case 1:} For $n=3$, $\ell=3$, so $f(x)=x^3$ is $3$-to-$1$. The $15$ nonzero elements form $5$ preimage sets of size $3$, yielding $|\mathrm{Im}(f)|=1+\frac{15}{3}=6$ with $\mathrm{Im}(f)=\{0,1,\beta^3,\beta^6,\beta^9,\beta^{12}\}$. The preimage sets are
	\begin{align*}
		\mathcal{U}_{0}
		&= \{(0,0)\} \mapsto 0,\\
		\mathcal{U}_{1}
		&= \{(1,0), (\alpha,0), (\alpha^2,0)\} 
		\leftrightarrow \{1, \beta^5, \beta^{10}\} \mapsto 1, \\[3pt]
		\mathcal{U}_{\beta}
		&= \{(0,1), (0,\alpha), (0,\alpha^2)\} \leftrightarrow \{\beta, \beta^6, \beta^{11}\} \mapsto \beta^3, \\[3pt]
		\mathcal{U}_{\beta^2}
		&= \{(\alpha,1), (\alpha^2,\alpha), (1,\alpha^2)\} \leftrightarrow \{\beta^2, \beta^7, \beta^{12}\} \mapsto \beta^6, \\[3pt]
		\mathcal{U}_{\beta^3}
		&= \{(\alpha,\alpha^2), (\alpha^2,1), (1,\alpha)\}  \leftrightarrow \{\beta^3, \beta^8, \beta^{13}\} \mapsto \beta^9, \\[3pt]
		\mathcal{U}_{\beta^4}
		&= \{(1,1), (\alpha,\alpha), (\alpha^2,\alpha^2)\}  \leftrightarrow \{\beta^4, \beta^9, \beta^{14}\} \mapsto \beta^{12}.
	\end{align*}
	\textbf{Case 2:} For $n=5$, $\ell=5$, so $f(x)=x^5$ is $5$-to-$1$. The $15$ nonzero elements form $3$ preimage sets of size $5$, yielding $|\mathrm{Im}(f)|=1+\frac{15}{5}=4$ with $\mathrm{Im}(f)=\{0,1,\beta^5,\beta^{10}\}$. The  preimage sets are
	\begin{align*}
		\mathcal{U}_{0}
		&= \{(0,0)\} \mapsto 0,\\
		\mathcal{U}_{1}
		&= \{(1,0), (\alpha,\alpha^2), (0,\alpha), (\alpha,\alpha), (1,\alpha^2)\} \leftrightarrow \{1, \beta^3, \beta^6, \beta^9, \beta^{12}\} \mapsto 1, \\[4pt]
		\mathcal{U}_{\beta}
		&= \{(0,1), (1,1), (\alpha^2,\alpha), (\alpha^2,0), (1,\alpha)\} \leftrightarrow \{\beta, \beta^4, \beta^7, \beta^{10}, \beta^{13}\} \mapsto \beta^5, \\[4pt]
		\mathcal{U}_{\beta^{2}}
		&= \{(\alpha,1), (\alpha,0), (\alpha^2,1), (0,\alpha^2), (\alpha^2,\alpha^2)\}  \leftrightarrow \{\beta^2, \beta^5, \beta^8, \beta^{11}, \beta^{14}\} \mapsto \beta^{10}.
	\end{align*}
\end{example}

\begin{example}
	\label{exm3}
	Consider the monomial mapping $f(x)=x^n$ over $\mathbb{F}_{3^2}$, where $|\mathbb{F}_{3^2}^{*}|=8$ and $\ell=\gcd(n,8)$. Let $\beta$ be a primitive element of $\mathbb{F}_{3^2}$, represent each field element as $(a,b)\in\mathbb{F}_3^2$.
	Both cases yield the same image set and preimage partition; only the image assignments to the preimage sets differ.
	
	\textbf{Case 1:} For $f(x)=x^2$, $\ell=\gcd(2,8)=2$. The preimage sets are
	\begin{align*} \mathcal{U}_{0} &= \{(0,0)\} \mapsto 0,\\ \mathcal{U}_{1} &= \{(1,0), (2,0)\} \leftrightarrow \{1,\beta^4\} \mapsto 1, \\ \mathcal{U}_{\beta} &= \{(0,1), (0,2)\} \leftrightarrow \{\beta,\beta^5\} \mapsto \beta^{2}, \\ \mathcal{U}_{\beta^{2}} &= \{(1,2), (2,1)\} \leftrightarrow \{\beta^2,\beta^6\} \mapsto \beta^{4}, \\ \mathcal{U}_{\beta^{3}} &= \{(2,2), (1,1)\} \leftrightarrow \{\beta^3,\beta^7\} \mapsto \beta^{6}. \end{align*}
	
	\textbf{Case 2:} For $f(x)=x^6$, $\ell=\gcd(6,8)=2$. The preimage sets are
	\begin{align*} \mathcal{U}_{0} &= \{(0,0)\} \mapsto 0,\\ \mathcal{U}_{1} &= \{(1,0), (2,0)\} \leftrightarrow \{1,\beta^4\} \mapsto 1, \\ \mathcal{U}_{\beta} &= \{(2,2), (1,1)\} \leftrightarrow \{\beta^3,\beta^7\} \mapsto \beta^{2}, \\ \mathcal{U}_{\beta^{2}} &= \{(1,2), (2,1)\} \leftrightarrow \{\beta^2,\beta^6\} \mapsto \beta^{4}, \\ \mathcal{U}_{\beta^{3}} &= \{(0,1), (0,2)\} \leftrightarrow \{\beta,\beta^5\} \mapsto \beta^{6}. \end{align*}

\end{example}

The following corollary states the simplified Plotkin-type bound on the optimal redundancy of FCCs for the monomial functions.

\begin{cor}
	Let $f(x)=x^n$ be a monomial function over $\mathbb{F}_{q^k}$, and let $\ell=\gcd(n,q^k-1)$. Then, the optimal redundancy $r_H^{m}(q,k,t)$ for the monomial function satisfies
	\begin{equation}\label{eq:plotkin_mon}
		r_H^{m}(q,k,t) \ge \frac{(2t+1)(q^k-1)(q^k-\ell+1)+(\ell-1)k(q-1)q^{k-1}}{q^{2k-1}(q-1)} - k.
	\end{equation}	
\end{cor}
\begin{IEEEproof}	
	For the monomial function $f(x)=x^n$ over $\mathbb{F}_{q^k}$, the zero preimage set is the singleton $\mathcal{U}_0=\{0\}$, and every nonzero preimage set has cardinality $\ell$. Furthermore, if $\beta$ is a primitive element of $\mathbb F_{q^k}$, the distinct nonzero preimage sets are given by Corollary \ref{cor_monomial}. Let $M=\frac{q^k-1}{\ell}$ be the number of nonzero preimage sets and the kernel of the mapping be $K = \ker(f) = \{\omega \in \mathbb{F}_{q^k}^{\ast}: \omega^n = 1\}$. Therefore, we have $\ell = |K| = \gcd(n, q^k-1)$. Since $\mathcal{U}_0=\{0\}$ has cardinality $1$, and each nonzero preimage set $\mathcal{U}_{\beta^g}$ for $0\le g\le M-1$ has cardinality $\ell$, the sum of the squares of the cardinalities of the preimage sets is given by
	\begin{equation}\label{eq:sum_squares_mon}
		\sum_{i} n_i^2 = 1 +M\ell^2 =1 + \frac{q^k-1}{\ell} \ell^2 =1+ \ell(q^k-1) .
	\end{equation}
	
	Now we proceed to show that the sum of pairwise distances among vectors within each preimage set $\displaystyle \sum_{i=1}^{E} \sum_{\mathbf{u},\mathbf{v} \in \mathcal{U}_{\alpha_{i}} } d_{H}(\mathbf{u}, \mathbf{v}) = (\ell-1)k(q-1)q^{k-1}$. Since $\mathcal{U}_0 = \{0\}$ contains only one element, its intra-set distance is zero, allowing us to restrict the total distance sum exclusively to the nonzero classes from $g=0$ to $M-1$. Thus, $\displaystyle \sum_{i=1}^{E} \sum_{\mathbf{u},\mathbf{v} \in \mathcal{U}_{\alpha_{i}} } d_{H}(\mathbf{u}, \mathbf{v})=\sum_{g=0}^{M-1}\sum_{\mathbf{u},\mathbf{v}\in \mathcal{U}_{\beta^g}} d_H(\mathbf{u},\mathbf{v})$. Using the identity $d_H(\mathbf{u},\mathbf{v})=w_H(\mathbf{u}-\mathbf{v})$ and writing each element of $\mathcal{U}_{\beta^g}$ as $\beta^g\omega$, we obtain $\displaystyle \sum_{g=0}^{M-1}\sum_{\mathbf{u},\mathbf{v}\in \mathcal{U}_{\beta^g}} d_H(\mathbf{u},\mathbf{v})=\sum_{\omega\in K}\sum_{a=0}^{q^k-2}w_H\bigl(\beta^a(\omega-1)\bigr)=\sum_{\substack{\omega\in K\\ \omega\neq1}}\sum_{a=0}^{q^k-2} w_H\bigl(\beta^a(\omega-1)\bigr)$. 
	
	The last equality follows since the term corresponding to $\omega=1$ contributes zero. For any fixed $\omega\neq1$, multiplication by $\omega-1$ is a bijection on $\mathbb F_{q^k}^{*}$, and hence $\displaystyle \sum_{a=0}^{q^k-2}w_H\bigl(\beta^a(\omega-1)\bigr)=\sum_{z\in\mathbb F_{q^k}^{*}} w_H(z)$. Because $|K|=\ell$ and the total weight of the non-zero field elements evaluates to $\displaystyle \sum_{z\in\mathbb F_{q^k}^{*}} w_H(z)=k(q-1)q^{k-1}$. Therefore, we have 
	\begin{equation}\label{eq:sum_squares_mon2}
		\displaystyle \sum_{i=1}^{E} \sum_{\mathbf{u},\mathbf{v} \in \mathcal{U}_{\alpha_{i}} } d_{H}(\mathbf{u}, \mathbf{v}) = (\ell-1)k(q-1)q^{k-1}.
	\end{equation}
	
	 Substituting \eqref{eq:sum_squares_mon} and \eqref{eq:sum_squares_mon2} into \eqref{eq:Hamming_plotkin}, we obain
	 \begin{align}
	 	r_H^{m}(q,k,t) & \ge \frac{1}{(q-1)q^{2k-1}} \left[ (2t+1) \left(q^{2k}- (1+ \ell(q^k-1)) \right)+ (\ell-1)k(q-1)q^{k-1} \right] - k \nonumber\\
	 	& = \frac{1}{(q-1)q^{2k-1}} \left[ (2t+1) \left(q^{2k}- 1- \ell(q^k-1) \right)+ (\ell-1)k(q-1)q^{k-1} \right] - k \nonumber\\
	 	& = \frac{1}{(q-1)q^{2k-1}} \left[ (2t+1) \left((q^k-1)(q^k+1)-\ell(q^k-1) \right)+ (\ell-1)k(q-1)q^{k-1} \right] - k \nonumber\\
	 	& = \frac{1}{(q-1)q^{2k-1}} \left[ (2t+1)(q^k-1)(q^k+1-\ell)+ (\ell-1)k(q-1)q^{k-1} \right] - k \nonumber.
	 \end{align}
\end{IEEEproof}

\begin{example}

For the monomial mapping $f(x)=x^5$ over $\mathbb{F}_{2^4}$, we have
$\ell=\gcd(5,16-1)=\gcd(5,15)=5=n$. Hence, the $q^k-1=15$ nonzero elements are partitioned into $M=\frac{q^k-1}{\ell}=3$ preimage sets, each of size $\ell=5$. The term $S = \sum_{g=0}^{M-1} \sum_{\mathbf{u},\mathbf{v} \in\mathcal{U}_{\beta^g}} d_H(\mathbf{u},\mathbf{v})$ represents the total Hamming distance within these preimage sets. One of the preimage sets is given by
$\mathcal{U}_{1}=\{1000, 0001, 0011, 0101, 1111\}$. The total Hamming distance obtained by summing, for each vector, its distance to all other vectors in the set is $44$. Repeating this computation for all three disjoint preimage sets yields $S=128$. Substituting this value into the redundancy bound in \eqref{eq:plotkin_mon} for t=1, we obtain $r\ge2$.

\end{example}

\section{Simplified Plotkin-type Bounds for functions under the Lee distance}\label{sec_Plotkin_Lee}	
In this section, we simplify the Plotkin-type bound for FCCs with the Lee distance in Corollary \ref{cor:Lee_plotkin} for several classes of functions by explicitly determining the sizes of the induced partition classes and the pairwise distances among vectors within each class. We consider the Lee weight function, the Lee weight distribution function and the modular sum function. 

\subsection{Lee weight function}
The Lee weight of a vector quantifies its total deviation from the zero vector under the Lee metric and serves as a natural measure of error magnitude. FCCs for the Lee weight function were previously studied in \cite{GA} and \cite{HRS}. The Lee weight function is given by  $f(\mathbf{u}) = w_L(\mathbf{u})$, where $\mathbf{u} \in \mathbb{Z}_q^k$. The expressiveness of this function is $E=k \left\lfloor \frac{q}{2} \right\rfloor + 1$, and $Im(f)=\{0,1,...,k \left\lfloor \frac{q}{2} \right\rfloor - 1, k \left\lfloor \frac{q}{2} \right\rfloor\}$.  For each $ i \in [0:k \left\lfloor \frac{q}{2} \right\rfloor ]$, we have $\mathcal{U}_{i} \triangleq \{\mathbf{u} \in \mathbb{F}_q^k : w_L(\mathbf{u})=i \}$. Next, we obtain an expression for $n_i=|\mathcal{U}_{i}|$. 

Let $p=\left\lfloor \frac{q}{2}\right\rfloor$. For each $j \in [p]$, let $w_j$ denote the number of coordinates having Lee weight $j$ in a given vector. Then a vector of Lee weight $i$ satisfies $\displaystyle \sum_{j=1}^{p} jw_j=i$. Furthermore, the number of coordinates having Lee weight $0$ is $\left(k-\displaystyle \sum_{j=1}^{p}w_j\right)$. Hence, for a fixed tuple $(w_1,\ldots,w_p)$ satisfying $\displaystyle \sum_{j=1}^{p} jw_j=i$, there are $\frac{k!}{w_1!\cdots w_p!\left(k-\sum_{j=1}^{p}w_j\right)!}$ ways to distribute the counts $w_1,\ldots,w_p$ among the $k$ positions. 

For odd $q$, each nonzero Lee weight $j \in [p]$ is attained by exactly two symbols in
$\mathbb{Z}_q$, namely $j$ and $q-j$. Hence, for a fixed tuple $(w_1,\ldots,w_p)$, there are $2^{\sum_{j=1}^{p}w_j}$ ways to assign symbols to the nonzero coordinates. Summing over all admissible tuples $(w_1,\ldots,w_p)$, we obtain $$n_i = \displaystyle
\sum_{\substack{(w_1,\ldots,w_p) : \\
		w_1+2w_2+\cdots+pw_p=i
}}
\frac{k!\,2^{\sum_{j=1}^{p}w_j}}
{w_1!\cdots w_p!
	\left(k-\sum_{j=1}^{p}w_j\right)!}. $$

For even $q$, every Lee weight $j\in [p-1]$ is attained by exactly the two symbols $j$ and $q-j$, where as the Lee weight $p=q/2$ is attained by the unique symbol $q/2$. Consequently, for a fixed tuple $(w_1,\ldots,w_p)$, there are $2^{\sum_{j=1}^{p-1}w_j}$ ways to assign symbols to the nonzero coordinates. Therefore, we obtain $$n_i = \displaystyle
\sum_{\substack{(w_1,\ldots,w_p ): \\
		w_1+2w_2+\cdots+pw_p=i
}}
\frac{k!\,2^{\sum_{j=1}^{p-1}w_j}}
{w_1!\cdots w_p!
	\left(k-\sum_{j=1}^{p}w_j\right)!}. $$

Thus, for the Lee weight function we have 
\begin{equation}
	\label{eq:ni_lee}
	n_i=
	\begin{cases}
		\displaystyle
		\sum_{\substack{(w_1,\ldots,w_p) : \\
				w_1+2w_2+\cdots+pw_p=i
		}}
		\frac{k!\,2^{\sum_{j=1}^{p}w_j}}
		{w_1!\cdots w_p!
			\left(k-\sum_{j=1}^{p}w_j\right)!},
		\hfill \text{ for odd } q,
		\\[3ex]
		\displaystyle
		\sum_{\substack{(w_1,\ldots,w_p) : \\
				w_1+2w_2+\cdots+pw_p=i
		}}
		\frac{k!\,2^{\sum_{j=1}^{p-1}w_j}}
		{w_1!\cdots w_p!
			\left(k-\sum_{j=1}^{p}w_j\right)!},
		\hfill \text{ for even } q.
	\end{cases}
\end{equation}

The expression for $n_i$ can be simplified for small alphabet sizes. In particular, explicit expressions for $q=3$ and $q=4$ are given below.

\noindent $\bullet$ $q=3:$ For $q=3$, we have $p=1$. Substituting $p=1$ in (\ref{eq:ni_lee}) for odd $q$, we obtain $n_i = \frac{k! 2^i}{i!(k-i)!}=2^i\binom{k}{i}$. \\
\noindent $\bullet$ $q=4:$ For $q=4$, we have $p=2$. Substituting $p=2$ in (\ref{eq:ni_lee}) for even $q$, we obtain $n_i = \displaystyle \sum_{\substack{(w_1,w_2) : w_1+2w_2=i }} \frac{k!2^{w_1}}{w_1!w_2!(k-w_1-w_2)!} $. Setting $w_2=j$, we have $w_1=i-2j$. Substituting these parameters, we obtain $n_i = \displaystyle \sum_{j=0}^{\lfloor \frac{i}{2} \rfloor}  \frac{k!2^{i-2j}}{(i-2j)!j!(k-i+j)!} $. The right-hand side expression is the coefficient of $x^i$ in the multinomial expansion of $(1+2x+x^2)^k$. Since $(1+2x+x^2)^k=(1+x)^{2k}$, from the binomial expansion of $(1+x)^{2k}$, we obtain the coefficient of $x^i$ as $\binom{2k}{i}$. Therefore, we have $n_i = \displaystyle \sum_{j=0}^{\lfloor \frac{i}{2} \rfloor}  \frac{k!2^{i-2j}}{(i-2j)!j!(k-i+j)!}$ $= \binom{2k}{i} $.
Thus, $n_i=\binom{k}{i}, 2^i\binom{k}{i}, \binom{2k}{i}$ for $q= 2, 3,$ and $4$, respectively.
The following example illustrates the computation of the cardinalities of the preimage sets using (\ref{eq:ni_lee}).
\begin{example} \label{ex1}
	Consider the Lee weight function on $\mathbb{Z}_5^2$. Since
	$p=\left\lfloor \frac{5}{2}\right\rfloor=2$, the odd $q$ expression in
	(\ref{eq:ni_lee}) gives 
	\[ 
	n_i=
	\sum_{\substack{(w_1, w_2):\\
			w_1+2w_2=i\\
	}}
	\frac{2!\,2^{w_1+w_2}}
	{w_1!\,w_2!\,(2-w_1-w_2)!},
	\qquad 0\le i\le 4.
	\]
	Evaluating the above expression gives
	$
	(n_0,n_1,n_2,n_3,n_4)=(1,4,8,8,4).
	$
	Hence, the partition induced by the Lee weight function consists of the preimage sets
	$
	\mathcal U_i=\{\mathbf u\in\mathbb Z_5^2:\mathrm{w}_L(\mathbf u)=i\},
	$
	$i\in\{0,1,\ldots,4\}$, with cardinalities $1,4,8,8$, and $4$, respectively. Moreover,
	$
	\sum_{w=0}^{4}n_i=25=|\mathbb Z_5^2|,
	$
	as expected.
\end{example}
Since $\allowbreak d_L(\mathbf u,\mathbf v) = \sum_{j=1}^{k} d_L(u_j,v_j)$, we have $\sum_{\mathbf u,\mathbf v\in\mathcal U_i} d_L(\mathbf u,\mathbf v) = \sum_{j=1}^{k} \sum_{\mathbf u,\mathbf v\in\mathcal U_i} d_L(u_j,v_j)$.

The Lee weight of a vector depends only on the sum of the Lee weights of its coordinates. Therefore,  the set $\mathcal U_i$ is invariant under coordinate permutations. Fix $j\in[k]$, and let $\pi$ be the permutation of $[k]$ defined by 
\[\small \pi(l)=
\begin{cases}
	\begin{aligned}
		&j, \text{ for }l = 1,  \\ 
		&1, \text{ for }l = j,  \\
		&l, \text{ for }l = [k] \setminus \{1, j\}. 
	\end{aligned}
\end{cases}
\] and $\pi(\mathbf u))$ is defined by,  $(\pi(\mathbf u))_j = u_{\pi(j)}$ $\text{ for all } j \in[k]$. For each $\alpha\in\mathbb Z_q$, define $A_1(\alpha)=\{\mathbf u\in\mathcal U_i:u_1=\alpha\},\,A_j(\alpha)=\{\mathbf u\in\mathcal U_i:u_j=\alpha\}$.

Define $\phi:A_1(\alpha)\to A_j(\alpha)$ by
$\phi(\mathbf u)=\pi(\mathbf u)$. Since
$(\pi(\mathbf u))_j=u_{\pi(j)}=u_1=\alpha$,
$\phi$ maps $A_1(\alpha)$ into $A_j(\alpha)$. As $\pi$ is invertible, $\phi$ is injective. For any $\mathbf v\in A_j(\alpha)$, invariance of $\mathcal U_w$ under coordinate permutations implies $\mathbf u=\pi^{-1}(\mathbf v)\in\mathcal U_w$, and
$
u_1=(\pi^{-1}(\mathbf v))_1=v_{\pi^{-1}(1)}=v_j=\alpha.
$
Thus $\mathbf u\in A_1(\alpha)$ and $\phi(\mathbf u)=\mathbf v$, proving surjectivity. Therefore $\phi$ is a bijection, and
$|A_1(\alpha)|=|A_j(\alpha)|$ for all $\alpha\in\mathbb Z_q$ and $j\in[k]$.
Thus the $j$th coordinate column is a permutation of the first for all $j\in[k]$. Hence,
$
\sum_{\mathbf u,\mathbf v\in\mathcal U_i}
d_L(u_j,v_j)
=
\sum_{\mathbf u,\mathbf v\in\mathcal U_i}
d_L(u_1,v_1),
\quad \forall\, j\in[k].
$	
Therefore, 
\begin{equation} \label{eq:sum_distance_simp} 
	\sum_{\mathbf u,\mathbf v\in\mathcal U_i}
	d_L(\mathbf u,\mathbf v)
	=
	k
	\sum_{\mathbf u,\mathbf v\in\mathcal U_i}
	d_L(u_1,v_1).
\end{equation}

Summing over all $i\in\{0,1,\ldots,kp\}$ yields the following simplification
\begin{equation} 
	\label{eq:lee_weight_symmetry}
	\sum_{i=0}^{kp}
	\sum_{\mathbf u,\mathbf v\in\mathcal U_i}
	d_L(\mathbf u,\mathbf v)
	=
	k
	\sum_{i=0}^{kp}
	\sum_{\mathbf u,\mathbf v\in\mathcal U_i}
	d_L(u_1,v_1).
\end{equation}
We next illustrate (\ref{eq:lee_weight_symmetry}) for the Lee weight function.
\begin{example}
	Consider the Lee weight function on $\mathbb{Z}_4^2$. The preimage set corresponding to Lee weight $2$ is 
	$\mathcal{U}_2=
	\{
	(2,0),(0,2),
	(1,1),(1,3),
	(3,1),(3,3)
	\}$, so that $|\mathcal{U}_2|=6$, which agrees with $n_i=\binom{2k}{i}$, since $n_2=\binom{4}{2}=6$. Direct computation gives $\sum_{\mathbf{u},\mathbf{v}\in\mathcal{U}_2}
	d_L(u_1,v_1)=36$.
	Since $k=2$, (\ref{eq:sum_distance_simp}) yields $\sum_{\mathbf{u},\mathbf{v}\in\mathcal{U}_2}
	d_L(\mathbf{u},\mathbf{v})
	=
	2\sum_{\mathbf{u},\mathbf{v}\in\mathcal{U}_2}
	d_L(u_1,v_1)
	=72$,
	which agrees with direct computation. This verifies (\ref{eq:sum_distance_simp}) for $\mathcal{U}_2$, and hence (\ref{eq:lee_weight_symmetry}) for $q=4$ and $k=2$.
\end{example}
Equations \eqref{eq:ni_lee} and \eqref{eq:lee_weight_symmetry} provide the values of $n_
i$ and the corresponding sums of pairwise Lee distances, respectively. Substituting these quantities into Corollary \ref{cor:Lee_plotkin} gives the following corollary.
\begin{cor}\label{cor:plotkin_lee_weight}
	The optimal redundancy $r_L^{\mathrm{w}_L}(q,k,t)$ for the Lee weight function satisfies
	\begin{equation}
		\label{eq:plotkin_lee_weight}
			r_L^{\mathrm{w}_L}(q,k,t)
			\ge 
			\frac{1}{Sq^{2k-1}}
			\left[
			(2t+1)
			\left(
			q^{2k}
			-\sum_{i=0}^{kp} n_i^2
			\right)
			+
			k
			\!\sum_{i=0}^{kp}
			\sum_{\mathbf{u},\mathbf{v}\in\mathcal{U}_i}
			\!\!\! d_L(u_1,v_1)
			\right]
			-k,
	\end{equation}
	where $n_i$ is given by (\ref{eq:ni_lee}) and $p=\left\lfloor \frac{q}{2}\right\rfloor$.
\end{cor}
A simplified Plotkin-type bound on the optimal redundancy of FCCs for the Lee weight function was derived in \cite{HRS} for $q \ge 5$ and $t=\lfloor \frac{q-3}{2} \rfloor$, while another lower bound was obtained in \cite{GA} for $k > \lceil \frac{t+1}{\lfloor \frac{q}{2} \rfloor} \rceil$. In contrast, the proposed bound in Corollary \ref{cor:plotkin_lee_weight} is applicable in all parameter regimes. 
\subsection{Lee weight distribution function}
The Lee weight distribution function is a fundamental function that captures the weight distribution of codewords.

\begin{defn}[Lee Weight Distribution Function \cite{GA}]
	The Lee weight distribution function is defined as  $f(\boldsymbol{u}) = \Delta_T(\boldsymbol{u}) \triangleq \left\lfloor \frac{\mathrm{w_L}(\boldsymbol{u})}{T} \right\rfloor$, where  $\boldsymbol{u} \in \mathbb{Z}_q^k$ and $T, k \in \mathbb{N}$. 
\end{defn}
The expressiveness of the Lee weight distribution function is
$E=\left\lceil \frac{k\left\lfloor \frac{q}{2} \right\rfloor+1}{T} \right\rceil$.
For each $j\in\{0,1,\ldots,E-1\}$,
define $\mathcal{V}_j
=\bigcup_{i=jT}^{\min\{(j+1)T-1,kp\}} \mathcal{U}_i$, where $\mathcal{U}_i
= \{\mathbf{u}\in\mathbb{Z}_q^k: \mathrm{w}_L(\mathbf{u})=i\}$ and $p=\lfloor \frac q2 \rfloor$. Since the sets $\{\mathcal{U}_i\}_{i=0}^{kp}$ are pairwise disjoint,
\begin{equation} \label{eq:card_preim_LWD} 
	N_j=|\mathcal{V}_j|=\sum_{i=jT}^{\min\{(j+1)T-1,kp\}}|\mathcal{U}_i|
	=\sum_{i=jT}^{\min\{(j+1)T-1,kp\}} n_i,
\end{equation}
where $n_i$ is given by (\ref{eq:ni_lee}). The following example computes the cardinalities of the preimage sets for the Lee weight distribution function.
\begin{example}
	Consider the Lee weight distribution function on $\mathbb Z_5^2$
	with $T=2$. Using the preimage cardinalities
	$(n_0,n_1,n_2,n_3,n_4)
	=
	(1,4,8,8,4)$
	computed in Example~\ref{ex1}, (\ref{eq:card_preim_LWD})
	gives $N_0=n_0+n_1=5, \,
	N_1=n_2+n_3=16,\, \text{ and }
	N_2=n_4=4$, respectively.
	Hence the induced partition sets have cardinalities $(N_0,N_1,N_2)=(5,16,4)$.
\end{example}

Fix $j\in\{0,1,\ldots,E-1\}$. Let $\{\mathcal{V}_j\}_{j=0}^{E-1}$ be the partition induced by the Lee weight distribution function. Since
$
d_L(\mathbf{u},\mathbf{v})
=
\sum_{l=1}^{k} d_L(u_l,v_l),
$
we have
$
\sum_{\mathbf{u},\mathbf{v}\in\mathcal{V}_j}
d_L(\mathbf{u},\mathbf{v})
=
\sum_{l=1}^{k}
\sum_{\mathbf{u},\mathbf{v}\in\mathcal{V}_j}
d_L(u_l,v_l).
$
For the Lee weight function, each set $\mathcal{U}_i=\{\mathbf{u}\in\mathbb{Z}_q^k:\mathrm{w}_L(\mathbf{u})=i\}$ is invariant under coordinate permutations, implying that the $l$th coordinate column is a permutation of the first column for all $l\in[k]$ (\ref{eq:sum_distance_simp}). Since
$
\mathcal{V}_j
=
\bigcup_{i=jT}^{\min\{(j+1)T-1,kp\}}
\mathcal{U}_i,
$
and the sets $\{\mathcal{U}_i\}_{i=0}^{kp}$ are pairwise disjoint, $\mathcal{V}_j$ is also invariant under coordinate permutations. Hence,
$
\sum_{\mathbf{u},\mathbf{v}\in\mathcal{V}_j}
d_L(u_l,v_l)
=
\sum_{\mathbf{u},\mathbf{v}\in\mathcal{V}_j}
d_L(u_1,v_1),
$
for all $l\in[k]$, and therefore
$
\sum_{\mathbf{u},\mathbf{v}\in\mathcal{V}_j}
d_L(\mathbf{u},\mathbf{v})
=
k
\sum_{\mathbf{u},\mathbf{v}\in\mathcal{V}_j}
d_L(u_1,v_1).
$
Summing over all $j$ establishes 
\begin{equation} \small
	\label{eq:lee_dist_symmetry}
	\sum_{j=0}^{E-1}
	\sum_{\mathbf{u},\mathbf{v}\in\mathcal{V}_j}
	d_L(\mathbf{u},\mathbf{v})
	=
	k
	\sum_{j=0}^{E-1}
	\sum_{\mathbf{u},\mathbf{v}\in\mathcal{V}_j}
	d_L(u_1,v_1).
\end{equation}

Equations \eqref{eq:card_preim_LWD} and \eqref{eq:lee_dist_symmetry} provide the values of $N_j$ and the corresponding sum of pairwise Lee distances within each set $\mathcal{V}_j$, respectively. Substituting these expressions into Corollary \ref{cor:Lee_plotkin} gives the following corollary.

\begin{cor}\label{cor:plotkin_lee_distribution}
	The optimal redundancy $r_L^{\Delta_T}(q,k,t)$ for the Lee weight distribution function satisfies
	\begin{equation}
		\label{eq:plotkin_lee_distribution}
			r_L^{\Delta_T}(q,k,t)
			\ge
			\frac{1}{Sq^{2k-1}}
			\left[
			(2t+1)
			\left(
			q^{2k}
			-\sum_{j=0}^{E-1}N_j^2
			\right)
			+
			k
			\sum_{j=0}^{E-1}
			\sum_{\mathbf{u},\mathbf{v}\in\mathcal{V}_j}
			\!\!d_L(u_1,v_1)
			\right]
			-k,
	\end{equation}
	where $N_j
	=
	\sum_{i=jT}^{\min\{(j+1)T-1,kp\}}
	n_i$, 
	and $n_i$ is given by (\ref{eq:ni_lee}).
\end{cor}

\subsection{Modular sum function}
We next consider the modular sum function, which computes the sum of the components of a vector modulo $q$. FCCs for this function were studied in \cite{HRS}. 
\begin{defn}[Modular Sum Function \cite{HRS}]
	The modular sum function $S_m:\mathbb{Z}_{q}^{k}\to\mathbb{Z}_{q}$ is defined by
	$
	S_m(\mathbf{u})
	=
	\left(\sum_{i=1}^{k}u_i\right)\bmod q
	$
	for $\mathbf{u}\in\mathbb{Z}_{q}^{k}$.
	Moreover, $\operatorname{Im}(S_m)=\mathbb{Z}_q$, and $|\operatorname{Im}(S_m)|=q$.
\end{defn}
We first show that the modular sum function is linear.
Let $\mathbf{x},\mathbf{y}\in\mathbb{Z}_q^k$ and
$\alpha,\beta\in\mathbb{Z}_q$. Then $\allowbreak S_m(\alpha\mathbf{x}+\beta\mathbf{y})
=\left(\sum_{i=1}^k(\alpha x_i+\beta y_i)\right)\bmod q
= \left(\alpha\sum_{i=1}^k x_i+\beta\sum_{i=1}^k y_i\right)\bmod q
=\alpha S_m(\mathbf{x})+\beta S_m(\mathbf{y})$.

Hence, $S_m$ satisfies the defining property of a linear function and is therefore linear. The following corollary is obtained by specializing the Plotkin-type bound for linear functions to the modular sum function.
\begin{cor}\label{bound_mod_sum}
	Let $S_m:\mathbb{Z}_q^k \to \mathbb{Z}_q$ be the modular sum function defined by
	$S_m(\mathbf{u})=\left(\sum_{i=1}^{k}u_i\right)\bmod q$. Then the optimal redundancy of an FCLC satisfies
	\begin{equation}
		r_L^{S_m} (q,k,t)
		\ge
		\frac{(q-1)(2t+1)}{S}
		+\frac{k}{q}
		-k.
	\end{equation}
\end{cor}

\begin{IEEEproof}
	Since the modular sum function is linear, the preimages of its values are precisely the cosets of
	$
	\ker(S_m)=\{\mathbf{u}\in\mathbb{Z}_q^k:\sum_{i=1}^{k}u_i\equiv0\pmod q\}
	$
	in $\mathbb{Z}_q^k$. Hence,
	$
	\mathcal{U}_{\alpha_i}=\mathbf{u}_i+\ker(S_m)
	$
	for some $\mathbf{u}_i$, and
	$
	n_i=|\mathcal{U}_{\alpha_i}|=q^{k-1}
	$
	for all $i$. Let
	$
	s=\sum_{\mathbf{w}\in\ker(S_m)}\mathrm{w}_L(\mathbf{w}).
	$
	Fix $j\in[k]$ and $a\in\mathbb{Z}_q$. The number of vectors in $\ker(S_m)$ with $u_j=a$ is $q^{k-2}$, since $k-2$ coordinates are arbitrary and the remaining coordinate is uniquely determined by the kernel condition. Thus, each symbol of $\mathbb{Z}_q$ appears exactly $q^{k-2}$ times in every coordinate position of $\ker(S_m)$. Using the additivity of the Lee weight,
	$s=\sum_{\mathbf{w}\in\ker(f)}\sum_{j=1}^{k}\mathrm{w}_L(w_j)
	=\sum_{j=1}^{k}q^{k-2}\sum_{a\in\mathbb{Z}_q}\mathrm{w}_L(a)
	=kq^{k-2}S$.
	
	Substituting $l=1$ and $s=kq^{k-2}S$ in (\ref{bound_linear_Lee}), we obtain the following explicit Plotkin-type bound for the modular sum function. 
	\small
	\begin{align} 
		r_L^{S_m}(q,k,t) \ge\;
		&\frac{(q-1)(2t+1)}{S_L}
		+\frac{k}{q}
		-k. \notag
	\end{align}
\end{IEEEproof}
The bound in Corollary~\ref{bound_mod_sum} is applicable for all values of $t$, whereas the simplified Plotkin-type bound for the modular sum function in \cite{HRS} requires $t\ge \frac{(\lfloor \frac q2 \rfloor-1)}{2}$ when $q$ is odd and $t\ge \frac{(\frac q2-1)}{2}$ when $q$ is even. Consequently, for larger alphabet sizes, the existing simplified bound is not applicable for small values of $t$.

\section{Conclusion}\label{sec_concl}
  In this paper, we considered FCCs with the Wyner-Graham distance framework, which includes the Hamming and Lee distances as special cases. We obtained a general Plotkin-type bound on the optimal redundancy of FCCs under Wyner–Graham distances, which depends only on the cardinalities of the preimage sets and the sum of pairwise distances among vectors within each preimage set. Compared to existing Plotkin-type bounds applicable to arbitrary functions, the proposed bound offers a computational advantage, as it requires fewer computations to evaluate. We then obtained a simplified Plotkin-type bound for linear functions under the Wyner-Graham distance. We further derived simplified bounds for specific functions of interest, such as the Hamming weight function, the Hamming weight distribution function, the Lee weight function, and the Lee weight distribution function, by characterizing the cardinalities of their preimage sets and the sums of pairwise distances among vectors within each preimage set. One direction of future research is to extend the approach developed in this paper to derive analogous bounds for FCCs under other distance measures, such as the $b$-symbol distance and the sum-rank distance, which do not fall within the Wyner-Graham distance framework.

\end{document}